\documentclass[runningheads]{llncs}

\usepackage[table,dvipsnames]{xcolor}

\usepackage{thmtools}
\newcounter{restatableTheorem}
\newcommand{\envheader}[2]{%
  \refstepcounter{restatableTheorem}\par\smallskip \textbf{{#2}~\thesection.\thetheorem.}
    \if\relax\detokenize{#1}\relax\else{({#1})}\fi 
}

\usepackage[colorinlistoftodos,prependcaption,textsize=small]{todonotes}

\usepackage[normalem,normalbf]{ulem}
\usepackage[font=small]{caption}
\usepackage{amsmath}
\usepackage{amssymb}
\usepackage{xspace}
\usepackage{tikz}
\usetikzlibrary{shapes.misc, positioning,decorations.pathreplacing}
\usetikzlibrary{arrows,automata,patterns}
\usetikzlibrary{arrows.meta,snakes}

\usepackage{subcaption}
\usepackage{wrapfig}
\usepackage{paralist}

\usepackage[inline]{enumitem}
\setlist[enumerate]{label = (\roman*), ref = \theenumii.\roman*}

\usepackage{mathtools}

\usepackage{leftindex}

\usepackage{rotating}
\usepackage{makecell}
\usepackage{enumitem}

\usepackage{tabularray}
\UseTblrLibrary{booktabs}
\usepackage{siunitx} 

\usepackage[linesnumbered,noend,ruled]{algorithm2e}
\DontPrintSemicolon
\SetKwInput{KwInput}{Input}
\SetKwInput{KwOutput}{Output}
\SetKw{Continue}{continue}
\SetKw{Break}{break}
\SetKwIF{If}{ElseIf}{Else}{if}{}{else if}{else}{}
\SetAlCapFnt{\small}
\SetAlCapNameFnt{\small}

\usepackage{orcidlink}
\usepackage{marvosym}
\usepackage{hyperref}
\usepackage[capitalize]{cleveref}
\Crefname{algocf}{Algorithm}{Algorithms}

\input{macros.tex}

\InputIfFileExists{cutmargins.tex}{%
\pdfpagesattr{/CropBox [125 82 490 688]} 
}{%
}

\title{String Solving with Stabilization and Transducers
\\(Technical Report)
}

\author{
  David Chocholatý\orcidlink{0009-0006-5614-1592}\inst{1} \and
  Vojtěch Havlena\orcidlink{0000-0003-4375-7954}\inst{1} \and
  Lukáš Holík~\href{mailto:lukasholik@cc.aau.dk}{\Letter}\orcidlink{0000-0001-6957-1651}\inst{1,2} \and\\
  Juraj Síč\orcidlink{0000-0001-7454-3751}\inst{1} \and
  Michal Šedý\orcidlink{0009-0004-0091-0546}\inst{2}
}

\institute{
  Faculty of Information Technology, Brno University of Technology, Brno, Czech Republic \\
  \email{\{ichocholaty, ihavlena, holik, sicjuraj\}@fit.vut.cz}
  \\ \and
  The Technical Faculty of IT and Design, Aalborg University, Aalborg, Denmark \\
  \email{\{lukasholik,medy\}@cc.aau.dk}
}

\begin{document}
\maketitle

\scaledvspace{-0.8cm}
\begin{abstract}
	We generalize an efficient automata-based approach to string constraint solving---the \emph{stabilization-based} method behind the solver \ziiinoodler---to support relational constraints represented by finite-state transducers (useful, for example, for modelling \replaceall constraints).
	We focus on an efficient treatment of length constraints by reducing the need for expensive concatenation elimination, which is a major bottleneck in automata-based string solving.
	We also propose powerful heuristics that significantly improve performance in practice.
	Implemented on top of \ziiinoodler, our method vastly outperforms existing solvers on benchmarks with relational constraints---it solves more instances and runs orders of magnitude faster.
\end{abstract}

\pagestyle{plain}

\scaledvspace{-0.5cm}
\section{Introduction}
\scaledvspace{-0.3cm}
String constraint solving has been studied intensively, with a number of competing approaches and solvers, e.g., \cite{AnthonyTowards2016,AnthonyReplaceAll2018,AnthonyComplex2019,ChainFree,VeanesDerivatives21,Z3str3RE,ganesh_expressive_23,BerzishDGKMMN23,cvc422,cvc5,AnthonyRegex2022,Flatten,DayEKMNP19,string_sat_23,RegularPropagationAnthonyJez25,Eisenhofer2026}.
As the field matures, it finds practical applications (e.g., analysis of resource access policies \cite{Rungta2022,stringsAWS18}) and approaches the point where it can fully fulfill its promises in program analysis and verification.
Automata-based approaches, originating from works such as \cite{Stranger,AutomataSplitting}, have always excelled in their generality, the range of handled constraints, and strong decidability guarantees (e.g., \cite{ChainFree,AnthonyComplex2019,AnthonyReplaceAll2018,RegularPropagationAnthonyJez25}).
Recent advances in the \emph{stabilization} algorithm \cite{NoodlerFM23,NoodlerOOPSLA23,NoodlerSAT24}, implemented in the solver \ziiinoodler~\cite{NoodlerTACAS24,NoodlerTACAS25}, have further elevated these approaches beyond the state-of-the-art in efficiency, as witnessed by the results of SMT-COMP in the last two years \cite{smtcomp24strings,smtcomp25strings} and by independent comparisons (e.g., \cite{hornstr,ostrich2,lotz2025s2s}).

In this paper, we extend the stabilization algorithm to handle \emph{transducer-represented relational constraints}, which are essential for modelling functions such as \replaceall, HTML encoding and decoding, character escaping and unescaping, and other transformations that arise in practical applications such as software analysis and policy verification.
We present a decision procedure for the \emph{chain-free fragment} of string formulae~\cite{ChainFree}, the largest practical fragment of string logic with transducers, equations, and length and regular constraints.
The main technical challenge we address is to efficiently integrate the stabilization algorithm with transducer constraints, and its combination with length constraints. In automata-based methods, length reasoning relies on eliminating concatenation,
a process that often causes severe combinatorial explosion due to the explicit enumeration of alignments/boundary interleaving of variables on opposing sides of equations/transducer constraints.
For example, a transducer constraint $T(xx,uv)$ or an equation $xx=uv$ would generate a case split of size 4, with cases having up to twice as many variables.
Thus, a sequence of such steps can generate an exponential number of cases, each exponentially larger in the number of variables and the length of the concatenation terms.
Stabilization is much cheaper, but its output cannot be used to resolve the length constraints.
We therefore propose a low-level combination of the two, which uses stabilization whenever possible and eliminates concatenation only from those sub-problems where interaction with the length constraints is unavoidable.
We further introduce several heuristics that make solving transducer constraints efficient in cases common in practice, which can
\begin{enumerate*}
	\item simplify certain ``homomorphic'' transducer constraints to a conjunction of simpler ones,
	\item replace transducer constraints with constant input/output by simpler equation constraints, and
	\item attempt to conclude (un)satisfiability of generally undecidable cyclic transducer constraints.
\end{enumerate*}

Note that our handling of negated constraints appears to be more general than in other takes on the automata approach.
From \cite{NoodlerOOPSLA23}, we take the ability to handle unrestricted negated equations.
We also observe that functional transducer constraints (the encoded relation is a function), covering nearly all known practical uses of transducers in string solving, can also be negated.
Specifically, a functional transducer constraint $T(x,y)$ can be rewritten as $T(x,y') \land y' \neq y$ where $y'$ is a fresh variable.
Although this is technically simple, to our knowledge, this observation does not appear to be widely known and has not been widely exploited in existing string solvers:
previous works treat transducers as non-negatable (complements of rational relations are generally not rational \cite{ChainFree,AnthonyTowards2016}). 

We have implemented our methods on top of the string solver \ziiinoodler and evaluated them across all available benchmarks that involve transducer constraints.
They outperform other solvers in almost all benchmark categories in the number of solved instances and run orders of magnitude faster than the second-best solver.

\paragraph{Related work.}
This work continues the branch \cite{NoodlerFM23,NoodlerFMJournal25,NoodlerOOPSLA23,NoodlerTACAS24,NoodlerSAT24,NoodlerTACAS25} of the automata approach \cite{AutomataSplitting} behind one of the most efficient string solvers \ziiinoodler, and adds an ability to handle transducer-definable constraints such as \replaceall. Our decision procedure handles the chain-free fragment \cite{ChainFree} of equations, transducers, regular, and length constraints, with a focus on an efficient combination with length constraints. Our approach is inspired by the combination of the stabilization method with elimination of concatenation from \cite{NoodlerOOPSLA23}, which does not consider transducers.
The automata approach was extended with transducers also in the branch of works \cite{AnthonyTowards2016,AnthonyComplex2019,AnthonyReplaceAll2018,AnthonyRegex2022,AnthonyInteger2020,RegularPropagationAnthonyJez25} behind the solver \ostrich \cite{ostrich2}. It handles the straight-line string logic, which is less expressive than chain-free fragment wrt. equations, regular constraints, lengths, and transducer constraints. On the other hand, they can handle additional types of constraints, including forms of \replaceall that we do not handle, e.g., with extended regular expressions in the replacement pattern.
%
The line of work was also recently extended to the chain-free fragment, called orderable in \cite{RegularPropagationAnthonyJez25}, but without focusing on efficient handling of length constraints.
%
Our method significantly outperforms these works on benchmarks we are aware of.
Transducers without length constraints were also considered in \cite{StringAFA}, and in \cite{Flatten}, experimenting with abstracting string constraints to Parikh images.
The lines of Z3Str* solvers \cite{Z3Str3,Z3Str4,Z3str3RE} handle the \replaceall function using
symbolic word equations combined with automata-based reasoning---progressing
from lazy decomposition
to regex-aware transducer semantics in \textsc{Z3strRE}
and
recursive sequence definitions.
\textsc{CVC4/5} apply a lazy model-refinement, unfolding
\replaceall into concatenation and regular constraints only when needed, guided by
context-dependent simplification \cite{cvc422}.
Earlier methods such as \cite{Stranger,S3,joxan-cav16}
also supported \replaceall,
relying on an automata-based over-approximation for PHP string analysis,
or used a progressive search over recursively defined string operations.
%

\section{Preliminaries}
\scaledvspace{-0.25cm}


\paragraph{Sets, strings and languages.}
$\numbersNatural$, $\mathbb{Z}$, and $\bool = \{ \top, \bot
	\}$ are sets of natural numbers (with~0), integers, and boolean values.
$\Sigma^*$ is a set of all finite sequences $w = a_1 \ldots a_n$ of \emph{symbols} from $\Sigma$,
called \emph{strings} or \emph{words} over $\Sigma$, with their \emph{length} defined as $|w| = n$ and $|\varepsilon| = 0$ for the \emph{empty string} $\varepsilon \in \Sigma^*$. $\SigmaEps = \Sigma \cup \{ \varepsilon \}$.
A \emph{concatenation} of strings $u$ and $v$ is denoted $u\concat v$ ($uv$ for short);
$\varepsilon$ is the neutral element.
We use the same notation for concatenating tuples.
We define an \emph{iteration} of a word $w$ as $w^0 = \varepsilon$ and $w^{k+1} = w^{k}\concat w$ for $k \in \numbersNatural$.
A \emph{language} $\lang$ over $\Sigma$ is a subset of $\SigmaStar$.
A \emph{concatenation} of languages is defined as $L_1\concat L_2 = \{ u\concat v \mid u\in L_1, v\in L_2 \}$,
where at the place of $L_1$ or $L_2$, we may write just $w$ to mean language $\{w\}$.

\scaledvspace{-2mm}
\paragraph{Relations.}
Given a tuple $\vecx$, $\vecx_i$ denotes its $i$-th item.
A \emph{projection} of an $n$-ary relation $\relation$ to the set of positions $I\subseteq\{1,\ldots n\}$ is $\projToSet{I}{\relation} = \{ (\vecx_i)_{i \in I} \,\mid\, \vecx \in \relation \}$ and we write $\projOutSet{j}{\relation}
$ to denote $\projection{\{1,\ldots,n\}\setminus \{j\}}$, the \emph{elimination} of a position $j$ from $R$.
The \emph{synchronization} of $R$ with an $m$-ary relation $\relation'$ at positions $i,j$
is the relation $\syncOnTapes{\relation}{\relation'}{i}{j} =
	\{\projOutSet{i}{\vecx}\cdot\vecy\mid \vecx\in \relation,\vecy\in \relation',\vecx_i=\vecy_j\}$.
The~\emph{composition} of relations is then
$\composeOnTapes{\relation}{\relation'}{i}{j} = \projOutSet{(n+j-1)}{\syncOnTapes{\relation}{\relation'}{i}{j}}$ (the synchronized position, the original $j$-th track of $\relation'$, is eliminated from the result).
When $i$ is the arity of $\relation$ and $j=1$ then we get the \emph{normal rational composition} $\compose{\relation'}{\relation}$.
For a set $S$, $\relationImg{\relation}{S} = \compose{S}{R}$
is the \emph{image} of $S$ under $\relation$.
We call the domain of a binary relation $\relation$ an input and the range of $\relation$ an output, and its \emph{inversion} is the relation $\preImg{\relation} = \{(b,a)\mid (a,b)\in \relation\}$.

\scaledvspace{-2mm}
\paragraph{Automata and transducers.}
A \emph{(nondeterministic) finite automaton} ((N)FA) over an alphabet $\Sigma$ is a tuple \mbox{$A=(Q,\Sigma,\delta,I,F)$} where
$Q$ is a finite set of states,
$\delta$ a set of transitions $\move q \aore r$ where $q,r\in Q$ and $\aore\in\Sigma$,
$I\subseteq Q$ a set of \emph{initial states}, and
$F\subseteq Q$ a set of \emph{final states}.
A~run of~$\aut$ over a~word~$w = a_1 \concats a_n \in \Sigma$
is a~sequence
$\move{p_0}{\aore_1}{p_1}
	\move{}{\aore_2}{}
	\ldots
	\move{}{\aore_n}{p_n}$ of transitions from $\Delta$.
The run is \emph{accepting} if $q_1 \in I$ and $q_n \in F$.
A \emph{language $\langOf{\aut}$ accepted by} $\aut$ is the set of all words for which $\aut$ has an accepting run.
\emph{Regular languages} are languages accepted by FAs.
An $n$-tape \emph{finite transducer} (FT) $\ft{T} $ over $\Sigma$ is a finite automaton over $\SigmaEpsTimes n$.
A~language $L$ over $\SigmaEpsTimes{n}$ represents the relation $\relationOf{L}\subseteq \SigmaEpsTimes{n}$ containing a tuple $(w_1,\ldots,w_n)$ iff there is a word $\veca^1 \cdots \veca^k \in L$ such that for all $1\leq i \leq n$, $w_i = \veca^1_i \concat \ldots \concat \veca_i^k$.
$\relationOf{L(\ft{T})}$ is the relation \emph{recognized} by the transducer.
The relations recognized by the transducers are called \emph{rational}.
If the relation recognized by a binary transducer is a \emph{function}, we call the transducer \emph{functional}.

\scaledvspace{-2mm}
\paragraph{Automata/transducer operations.}
We use several non-standard auto\-mata/trans\-ducer operations.
First, given two
FAs $\aut= (Q,\Sigma,\Delta,I,F)$, $\aut'= (Q',\Sigma,\Delta',I',F')$ and a delimiter symbol $\sharp$,
we construct their \emph{$\sharp$-concatenation}
$\aut \concat \sharp \concat \aut' = (Q \uplus Q', \Delta \uplus \Delta' \uplus \{ \move{p}{\sharp}{q} \mid p \in F, q \in I' \}, I, F')$ that accepts the language $L(A)\concat\sharp\concat L(A')$.

Second, the relational synchronization is implemented with transducers as the standard product construction:
For an $n$-tape transducer $T = (Q,\SigmaEpsTimes n,\delta,I,F)$ and an $m$-tape transducer $\ft{T}' = (Q',\SigmaEpsTimes n,\delta',I',F')$,
$\syncOnTapes{\ft{T}}{\ft{T}'}{i}{j}$ is the $(n+m-1)$-tape transducer
$(Q\times Q', \SigmaEpsTimes n, \delta^{\times}, I\times I', F\times F')$ where
$\delta^{\times}$ contains a transition $\ftTransition{(p, p')}{{\projOutSet{i}{\vecx}\cdot\vecy}}{(q, q')}$ in the following three cases:
\begin{enumerate*}
	\item $\ftTransition{p}{\vecx}{q} \in \delta$ with $\vecx_i=\varepsilon$, $\vecy=\varepsilon^m$, and $p' = q'$, or
	\item $\ftTransition{p'}{\vecy}{q'} \in \delta'$ with $\vecy_j=\varepsilon$, $\vecx=\varepsilon^n$, and $p = q$, or
	\item $\ftTransition{p}{\vecx}{q} \in \delta$, $\ftTransition{p'}{\vecy}{q'} \in \delta'$, $\vecx_i = \vecy_j \in \SigmaEps$.
\end{enumerate*}
Elimination $\projOut{J}$ is implemented simply by erasing the positions $J$ of each symbol in every transition.
The relational composition is then implemented as $\composeOnTapes{\ft{T}}{\ft{T}'}{i}{j} = \projOut{(n+m-1)}(\syncOnTapes{\ft{T}}{\ft{T}'}{i}{j})$.
Given a regular language $S$, we lift the relational image to the transducers as $\relationApplyOn{\ft{T}}{S} = \relationImg{\relationOf{\langOf{\ft{T}}}}{S}$.
The inverse $T^{-1}$ of a binary transducer $T$ is implemented by swapping the symbols of the two tapes on all transitions.

Third, we use \emph{transducers with delimiters} to delimit segments of accepted words.
A \emph{transducer over $\Sigma$ with delimiters} $\markers$ (\emph{$\markers$-delimited transducer}), $\markers\cap\Sigma = \emptyset$, is an automaton over $\SigmaEpsTimes n \cup \{\sharp^n\mid \sharp\in \markers\}$: if a delimiter $\sharp\in \markers$ appears on one tape of a transition, it is on all its tapes.
\emph{Synchronization of $\markers$-delimited transducers},
$\syncOnTapesMarked{\ft{T}}{\ft{T}'}{i}{j}{\markers}$, preserves the delimiters of both transducers. It is defined as $\syncOnTapes{\ft{T}}{\ft{T}'}{i}{j}$ above, but with two additional cases that generate its transitions:
\begin{enumerate*}[start=4]
	\item $\ftTransition{p}{\vecx}{q} \in \delta$ with $\vecx = \sharp^n$, $\vecy=\sharp^m$, $\sharp \in \markers$, and $p' = q'$, or
	\item $\ftTransition{p'}{\vecy}{q'} \in \delta'$ with $\vecy = \sharp^m$, $\vecx=\sharp^n$, $\sharp \in \markers$, and $p = q$.
\end{enumerate*}

\scaledvspace{-2mm}
\paragraph{String constraints.}
We use a minimal string constraint language built from \emph{atomic string constraints}
over an alphabet~$\Sigma$, string variables~$\vars$, and integer variables $\varslia$.
We define an \emph{equation} $\equationNodeArgs{t_s}{t_s'}$, a \emph{rational/transducer constraint} $T(t_s,t_s')$,
a \emph{regular constraint} $t_s \in \regex$,
and a length constraint $t_i \leq t_i'$ where
$t_s,t_s'\in \vars^*$ are \emph{string terms}, $t_i,t_i'$ linear integer arithmetic (LIA) terms over variables $\varslia\cup\{\lenvar{x}\mid x\in\vars\}$,
$\regex$ a regular language represented by an FA,
and $\replacement$ a rational relation represented by a transducer.
For $\equationNodeArgs{t_s}{t_s'}$ and $\transductionArgs{t_s}{t_s'}$,
$t_s$ is the \emph{input} (term) and $t_s'$ the \emph{output} (term), also called the \emph{left} and \emph{right} side/terms, respectively.
In this paper, the orientation of the equations and the transducer constraints is significant, and they are not commutative.
We identify automata with the accepted languages and transducers with the recognized relations.
%
%
A general \emph{string constraint} $\varphi$ over $\Sigma$ and $\vars$ is a boolean combination of atomic constraints.
$\varsOf{\varphi} \subseteq \vars$ denotes the set of string variables occurring in~$\varphi$; additionally, if $\lenvar{x}$ occurs in $\varphi$, then $x \in \varsOf{\varphi}$ and we call $x$ a \emph{length variable}.

An \emph{assignment} $\ass = \ass_s \cup \ass_i$ consists of a \emph{string assignment} $\ass_s \colon \vars \to \Sigma^*$ and an \emph{integer assignment} $\ass_i \colon \varslia \to \mathbb{Z}$.
\emph{Satisfaction} of a formula $\varphi$ by an assignment $\ass = \ass_s \cup \ass_i$, written $\ass\models\varphi$, is defined in the usual inductive way for boolean connectives from atomic expressions:
$\ass \models \equationNodeArgs{\sterm}{\tterm}$ iff
$\ass_s(\sterm) = \ass_s(\tterm)$ where $\ass_s(\tterm')$ for a term $\tterm' =
	x_1\ldots x_n$ is defined as $\ass_s(x_1)\cdot\ldots\cdot\ass_s(x_n)$,
$ \ass \models x\in \regex$ iff $\ass_s(x)\in \regex$,
$ \ass \models T(s,t)$ iff $(\ass_s(x),\ass_s(y))\in T$,
and $\ass \models t_i \leq t_i'$ iff $\ass'(t_i) \leq \ass'(t_i')$ where $\ass' = \ass_i \cup \{\lenvar{x} \mapsto |\ass_s(x)| \mid x\in\vars\}$. 
We call~$\formula$ \emph{satisfiable} if it has a \emph{solution}, i.e., a satisfying assignment.

A formula $\varphi$ is an \emph{existential extension} of $\psi$ iff $\varphi$ becomes equivalent to $\psi$ after existentially quantifying variables not appearing in $\psi$.
A \emph{cube} is a conjunction of atoms and negated atoms.
A \emph{positive cube} $\cube$ has no negation.
For $\cube$, we use a set of equations $\Eqof\cube$, regular constraints $\Regof\cube$, transducer constraints $\Trof\cube$, length constraints $\Lenof\cube$, equations and transducer constraints together $\EqTrof\cube$, and all constraints except length constraints $\Strof\cube$.
We abuse notation to identify positive cubes with sets of their constraints.
%
Without loss of generality, we assume cubes with \emph{normalized regular constraints} with exactly one regular constraint $x\in \regex_x$ per string variable $x\in \varsOf{\formula}$ (all regular constraints $x \in \regex_i$ can be replaced by $x \in \bigcap_{i=0}^n \regex_i$, obtained by the standard FA intersection construction).
We abuse the notation and use $\Regof\cube(x)$ to denote $\regex_x$ and for a string term $t_s = x_1\concats x_n$, we define $\Regof\cube(t_s) = \Regof\cube(x_1) \concats \Regof\cube(x_n)$.
When clear from the context, we may omit the subscript $\cube$.

\scaledvspace{-2mm}
\paragraph{Transducer constraints as \replaceall.}
In practice, the string solver $\ziiinoodler$ supports the rich string logic of SMT-LIB standard library~\cite{SMTLIB-Strings}.
The new ability to handle transducer constraints allows $\ziiinoodler$ to use atomic string constraints of the form
$t' = \replaceall(t, L, w)$ 
%
where $t,t'\in(\Sigma\cup\vars)^*$ are the \emph{subject} and the \emph{result}, respectively,
$L\subseteq\Sigma^*$ is the \emph{pattern}
(represented by a regular expression),
and $w\in\Sigma^*$ is the \emph{replacement}.
Under a string valuation $\ass$, $\replaceall(t, L, w)$
returns the string obtained by replacing, left-to-right, each
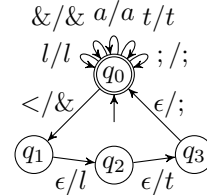
\begin{wrapfigure}[10]{r}{0.25\textwidth}
	\vspace{-9.5mm}
	\hspace{-1mm}
	\centering
	\begin{tikzpicture}[>=stealth',shorten >=0pt,auto,node distance=15mm,scale=0.8,every state/.style={minimum size=5mm,inner sep=1pt}]
  \node[state, accepting, initial below, initial text={}] (q0) {$q_0$};
  \node[state, below left of=q0] (q1) {$q_1$};
  \node[state, below right of=q0] (q3) {$q_3$};
  \node[state, below of=q0, node distance=12mm] (q2) {$q_2$};

  \draw[->] (q0) edge node[left,pos=0.3]{${<}/\&$} (q1);
  \draw[->] (q1) edge node[left,pos=0.5,yshift=-3mm,xshift=3mm]{$\epsilon/l$} (q2);
  \draw[->] (q2) edge node[right,pos=0.5,yshift=-3mm,xshift=-3mm]{$\epsilon/t$} (q3);
  \draw[->] (q3) edge[left] node[right,pos=0.7]{$\epsilon/;$} (q0);

  \draw[->] (q0) edge[loop, in=100, out=80, looseness=9] node[pos=0.5, above]{$a/a$} (q0);
  \draw[->] (q0) edge[loop, in=130, out=110, looseness=9] node[pos=0.5, above left]{$\&/\&$} (q0);
  \draw[->] (q0) edge[loop, in=160, out=140, looseness=9] node[pos=0.5, left]{$l/l$} (q0);
  \draw[->] (q0) edge[loop, in=50, out=70, looseness=9] node[pos=0.5, above right]{$t/t$} (q0);
  \draw[->] (q0) edge[loop, in=20, out=40, looseness=9] node[pos=0.5, right]{$;/;$} (q0);
\end{tikzpicture}

	\vspace{-0.25cm}
	\caption{Transducer $T$ for $vz = \texttt{replace-}$ $\texttt{All}(xy, <, \text{\&lt;})$.%
	}
	\label{fig:replaceall}
\end{wrapfigure}
shortest non-empty match of $L$ in $\ass(t)$ by $w$.
We note that if $L = \{w'\}$ is a singleton language, we abuse notation and write $\replaceall(t, w, w')$.
\replaceall
operations can also be  \emph{nested} (the subject $t$ can itself be an application of a \replaceall operation), allowing, for instance, to specify character encoding by its HTML codes.
We use the techniques of \cite{ReplaceAllToTransducer} to implement the translation of the \replaceall functions into transducer constraints.

\scaledvspace{-2mm}
\beginexample
\label{example:modelling-replace-operations}
Consider a string constraint $vz = \replaceall(xy, <, \text{\&lt;})$ over an alphabet $\Sigma = \{ a, <, \&, l, t, ; \}$ with string variables $x$, $y$, $v$, $z$.
It replaces all occurrences of $<$ in $xy$ by $\text{\&lt;}$ and stores the result into $vz$.
This is a common pattern occurring in the sanitization of user inputs in web applications when using the HTML escaping function of
various programming languages. In our approach, we model this string constraint using the transducer constraint $T(xy,vz)$ where $T$
is given in~\cref{fig:replaceall}.
\sectionEndSymbol

\scaledvspace{-0.5cm}
\section{Chain-free String Constraints} \label{sec:chain-free}
\scaledvspace{-0.25cm}

A string theory solver in the CDCL(T) framework solves cubes.
We concentrate on the chain-free fragment of cubes~\cite{ChainFree}.
We first define positive chain-free cubes, and later discuss how negative atoms can be reduced to a positive chain-free cube.

\scaledvspace{-2mm}
\paragraph{Positive chain-free cubes.}
We adapt the definition of chain-freeness from \cite{NoodlerFM23}, which is simpler than the original \cite{ChainFree}, and extend it with transducer constraints.
The definition is based on the notion of an \emph{inclusion graph} $G_\formula = (V,E)$ of a positive cube $\cube$, which is also an integral part of our decision procedure.
The nodes $V$ of the inclusion graph are the equations and the transducer constraints of $\cube$.
The graph has an edge $(c, c')\in E$ iff there is a shared variable between the output term of $c$ and the input term of $c'$.
\begin{figure}[t]
	\begin{subfigure}[b]{0.49\linewidth}
		\centering
		\begin{tikzpicture}[>=stealth',shorten >=0pt,auto,node distance=20mm,scale=0.8,baseline=(v<u.base)]
    \tikzset{every node/.style={rounded corners=1mm}}

    \node[draw] (u<z) {$\equationNodeArgs{z}{u}$};
    \node[draw,right of=u<z] (v<u) {$\equationNodeArgs{u}{v}$};
    \node[draw,right of=v<u] (xy<uvwr) {$\transs^{-1}(uvwr, xy)$};

    \draw
        (u<z) edge[->] (v<u)
        (u<z) edge[->,out=-20,in=-160] (xy<uvwr)
        (v<u) edge[->] (xy<uvwr);

\end{tikzpicture}
		\caption{Inclusion graph for $\formula_1$.}
		\label{fig:constraint_graph_examples-acyclic}
	\end{subfigure}
	\hfill
	\begin{subfigure}[b]{0.49\linewidth}
		\centering
		\begin{tikzpicture}[>=stealth',shorten >=0pt,auto,node distance=20mm,scale=0.8,baseline=(v<u.base)]
    \tikzset{every node/.style={rounded corners=1mm}}



    \node[draw] (u<z) {$\equationNodeArgs{z}{u}$};
    \node[draw,right of=u<z] (v<u) {$\equationNodeArgs{u}{v}$};
    \node[draw,right of=v<u] (yz<xuv) {$\transs(xuv, yz)$};

    \draw 
        (u<z) edge[->] (v<u)
        (u<z) edge[->,out=-20,in=-160] (yz<xuv)
        (u<z) edge[->,out=-20,in=-160] (yz<xuv)
        (u<z) edge[<-,out=20,in=160] (yz<xuv)
        (v<u) edge[->] (yz<xuv)
	;

\end{tikzpicture}
		\caption{Inclusion graph for $\formula_2$.}
		\label{fig:constraint_graph_examples-cyclic}
	\end{subfigure}
	\caption{
		Examples of inclusion graphs.
		(\subref{fig:constraint_graph_examples-acyclic})
		An acyclic inclusion graph for a chain-free $\formula_1\colon \equationNodeArgs{z}{u} \wedge \equationNodeArgs{u}{v} \wedge \replacementArgs{xy}{uvwr}$.
		(\subref{fig:constraint_graph_examples-cyclic})
		A~cyclic inclusion graph for non chain-free $\formula_2\colon \equationNodeArgs{z}{u} \wedge \equationNodeArgs{u}{v} \wedge \replacementArgs{xuv}{yz}$.
	}
	\label{fig:constraint_graph_examples}
	\vspace{-5mm}
\end{figure}

We say a positive cube $\cube$ is \emph{directly chain-free} iff
\begin{enumerate*}
	\item its inclusion graph is acyclic, and
	\item no variable occurs multiple times in the inputs of constraints (it can occur in only one input, with at most one occurrence).
\end{enumerate*}
%
Note that any chain-free cube can be transformed into an equivalent directly chain-free cube~\cite{NoodlerFM23}.
From now on, we presume that all chain-free cubes are in the directly chain-free form.
We define a \emph{dual variant} of a formula as a formula obtained by replacing some constraints by their duals,
where a \emph{dual} of an equation $\dual{\equationNodeArgs{s}{t}} = \equationNodeArgs{t}{s}$
and a dual of a transducer constraint $\dual{\transductionArgs{s}{t}} = \transductionArgsInv{t}{s}$.
We call a positive cube \emph{chain-free} iff it has a directly chain-free dual variant.
\cref{fig:constraint_graph_examples} shows an example of a (non-)chain-free cube and its graph.

We note that our definition of the inclusion graph differs from that in \cite{NoodlerFM23}, especially in the cyclic case: nodes on a cycle would have their duals in the graph as well.
The difference is related to technicalities and the presentation of the decision procedure.
It does not however influence the notion of chain-freeness.

\begin{restatable}{lemma}{lemmaChainFreeEquivalence}\label{lemma:chain-free-equivalence}
	Our definition of a chain-free cube is equivalent to the definition of a positive chain-free cube in~\cite{ChainFree}.
\end{restatable}

\paragraph{Chain-free formulae with negation.}
Chain-freeness can absorb a broad class of negated constraints by transforming them into positive constraints without affecting chain-freeness.
Negation of regular constraints can be eliminated by standard automata complementation.
To handle disequations (negated equations), we employ the approach of \cite{NoodlerSAT24}, relying on encoding disequations using $\tocode$ functions returning the Unicode code-point of a character.
We convert a disequation $\disequationNodeArgs{z_1\concats z_i}{\bar{z}_1\concats\bar{z}_j}$ to $\lenvar{z_1} + \concats + \lenvar{z_i} \neq \lenvar{\bar{z}_1} + \concats +\lenvar{\bar{z}_j}
	\vee ( \equationNodeArgs{z_1\concats z_i}{x_1a_1x_2} \land \equationNodeArgs{\bar{z}_1\concats\bar{z}_j}{y_1a_2y_2} \land \lenvar{x_1} = \lenvar{y_1} \land a_1 \in\Sigma \land a_2 \in \Sigma \land \tocode(a_1) \neq \tocode(a_2))$.
Since $\tocode(a_1) \neq \tocode(a_2)$ cannot generate equations that break chain-freeness, it is used instead of the more natural $\equationNodeArgs{a_1}{a_2}$.
$\tocode(a_1) \neq \tocode(a_2)$ is handled by a LIA solver after $\tocode(a)$ is transformed into a LIA constraint~\cite{NoodlerSAT24}.
Elimination of negation from general transducer constraints by complementation is impossible (complements of rational relations are generally not rational~\cite{ChainFree,AnthonyTowards2016}).
Hence, the chain-free formulae cannot contain negated general transducer constraints.
However, we notice that negation can be eliminated from \emph{functional transducers} (the encoded relation is a function): $\neg T(x,y)$ can be converted to an equivalent $T(x,z) \land \disequationNodeArgs{z}{y}$.
Both transformations for negated constraints preserve chain-freeness, i.e., a positive chain-free formula remains chain-free after adding the transformed version of any negated constraint.
Thus, negation can be easily absorbed by chain-freeness, provided that the negated transducers are functional.
Hence, we can extend the decidability result from \cite{ChainFree} with the more general handling of negation.

\begin{restatable}{theorem}{lemmaChainFreeWithNegationsAndTransducers}\label{lemma:chain-free-with-negations-and-transducers}
	Satisfiability of a cube where negated transducer constraints are functional and where the positive constraints are chain-free is decidable.
\end{restatable}

\section{Combining Stable and \ConcatenationFree Form} \label{sec:stable-monadic-forms}

Conceptually, we decide a positive cube $\cube$ in two steps. First, we derive a length image of $\Strof\cube$ (the part without length constraints), in the form of a LIA formula $\lenimg{\Strof\cube}$.
Second, we solve the LIA formula $\lenimg{\Strof\cube} \land \Lenof\cube$ using a LIA solver.

The \emph{length image} $\lenimg{\psi}$ of a formula $\psi$ is a LIA formula over variables $\{\lenvar x \mid  x\in \varsOf \psi\}$ that characterizes the lengths of the solutions of $\psi$: it has a solution $\nu$ iff $\psi$ has a solution $\mu$ where for each $x\in\varsOf \psi$, $\nu(\lenvar x) = |\mu(x)|$.

The generation of a length image $\lenimg \psi$ of a length-free formula goes back to \cite{AutomataSplitting,AnthonyTowards2016,ChainFree} and requires a \concatenationFree form with acyclic transducer constraints, possibly with trivial (solved) equations that do not influence satisfiability.
The \concatenationFree form is expensive. Its derivation requires generating alignments of variables on different sides, with variable splitting. That may, in theory, add two additional levels of exponential explosion relative to the stabilization of \cite{NoodlerFM23}.
Therefore, we combine the stable and \concatenationFree forms: the \concatenationFree form only for the part relevant to the length constraints and the cheaper stable form for the rest.
%
%
We will first discuss the two pure forms and then define the combination.
\paragraph{Stable form.}
A directly chain-free cube $\cube$ is \emph{stable} when
for each equation $\equationNodeArgs{s}{t}$, $\Regof\cube(s) \supseteq \Regof\cube(t)$, and for each transducer constraint $\transductionArgsOf{\ft{T}}{s}{t}$, $\relationApplyOn{\ft{T}}{\Regof\cube(s)} \supseteq \Regof\cube(t)$.
Stability together with the feasibility of regular constraints implies satisfiability, similar to \cite{NoodlerFM23,NoodlerOOPSLA23}, where transducer constraints are not considered.
%
%
\begin{restatable}{theorem}{lemmaStableHasSolution}\label{lemma:stable-has-solution}
	A stable cube with satisfiable regular constraints is satisfiable.
\end{restatable}
%

There is a major difference between the implications of stability for transducer constraints and for equations.
For only equations, a stronger version of stability implies satisfiability even in the non-chain-free case: the inclusions must also hold for the duals of constraints on the cycle; there are also additional conditions on multiple occurrences of variables in the inputs~\cite{NoodlerFM23}.
This no longer works for transducers, as witnessed by a cube $\ft{T}(x,x) \land x \in \{ a, b \}$ where $R(L(\ft{T})) = \{(a, b), (b, a)\}$.
The cube is stable, has unconstrained and hence non-empty languages of variables, with both inclusions holding true ($T[\Sigma^*] = T^{-1}[\Sigma^*] = \Sigma^*$), but it is unsatisfiable.
Thus, in this paper, we define stability only for the chain-free formulae.


%
%

\paragraph{\ConcatenationFree form.}
We call a positive cube $\cube$ \emph{concatenation-free} iff
\begin{enumerate*}
	\item concatenation appears only in \emph{solved equations} (also called \emph{binding equations}) of the form $\equationNodeArgs{x}{t_x}$ where $x$ has no other occurrence in $\cube$ and $\Regof\cube(x) \supseteq \Regof\cube(t_x)$, and
	\item the transducer constraints are chain-free (due to concatenation-freeness, it just means acyclic).
\end{enumerate*}

The inclusions satisfied by the binding equations mean $\Regof\cube(x)$ does not influence the satisfiability and the values of the variables of $t_x$, and therefore of any other variables of $\cube$, although the value of $t_x$ may still put additional constraints on the value of $x$.

Using the techniques discussed in \cite{AnthonyTowards2016,ChainFree} and in \cref{sec:lengths}, a length image can be extracted from a \concatenationFree cube by composing the transducer constraints into a single multi-tape transducer, extracting its Parikh image, and then equating the length of $x$ with the sum of the length of variables in $t_x$ for every binding equation.


\paragraph{Stable-concatenation-free form.}
We want a combined form $\cube$ where $\EqTrof\cube$ is split into a stable $\stablecOf{\cube}$ and a \concatenationFree $\solvedcOf{\cube}$ part, and the satisfiability depends only on the length image of $\solvedcOf{\cube}\land\solvedcReg$ in conjunction with $\Lenof\cube$, where $\solvedcReg{}$ is $\Regof{\cube}$ restricted to $\varsOf{\solvedc \land \Lenof{\cube}}$ and $\stablecOf{\cube}\land\Regof{\cube}$ is only required to be stable and can be further ignored.
$\stablecOf{\cube}\land\Regof\cube$ and $\solvedcOf{\cube}\land\solvedcReg\land\Lenof\cube$ are interfaced through regular constraints on the shared variables $\sharedv = \varsof{\stablecOf{\cube}}\cap\big(\varsof{\solvedcOf{\cube}}\cup\varsof{\Lenof\cube}\big)$.
Hence, we need $\stablecOf{\cube}$ to not impose any constraints on $\sharedv$ other than $\Regof\cube$: shared variables must remain unconstrained, so that length reasoning depends only on the concatenation-free part.
Thus, $\Regof\cube$ must imply $\stablecOf{\cube}$ projected to $\sharedv$ (the \emph{projection} means quantifying out all variables other than $\sharedv$).
In other words, that every assignment $\ass$ to $\sharedv$ with $\ass(x)\in\Regof\cube(x)$ for all $x\in \sharedv$ is a part of some solution of $\sharedv$.
In stable form, this holds exactly for \emph{output-only} variables---variables not appearing in the input terms.

\begin{restatable}{lemma}{stableProjection}\label{lemma:stable-projection}
	$\Regof\cube$ implies the projection of a stable cube $\cube$ to its output-only variables.
\end{restatable}

We now say that a cube $\varphi$ is in the \emph{stable-\concatenationFree} form with the split $\EqTrof\cube = (\stablecOf{\cube}\land\solvedcOf{\cube})$ if
$\stablecOf{\cube}\land\Regof\cube$ is stable,
$\solvedcOf{\cube}\land\Regof\cube$ is \concatenationFree, and
the variables $\varsof{\stablecOf{\cube}}\cap\big(\varsof{\solvedcOf{\cube}}\cup\varsof{\Lenof\cube}\big)$ do not appear in the inputs of $\stablecOf{\cube}$.

\begin{restatable}{theorem}{stableConcatenationFree}\label{lemma:stable-solved}
	A stable-\concatenationFree cube $\cube$ with the split $\EqTrof\cube = (\stablecOf{\cube}\land\solvedcOf{\cube})$ is an existential extension of
	$\solvedcOf{\cube} \land \solvedcReg \land \Lenof\cube$.
\end{restatable}

Transformation to stable-\concatenationFree form, maximizing the stable and minimizing the \concatenationFree part, is the core of our decision procedure in \cref{sec:dec-proc}.
The split is not performed at the level of the input formula, but rather when variables are split during the elimination of concatenation in the decision procedure that introduces new variables. The new variables make it possible to separate the length dependent part more effectively, as illustrated in \cref{example:stabilvsconcatelim} below.
Although not required by \cref{lemma:stable-solved}, our decision procedure in \cref{sec:dec-proc} derives the \concatenationFree part in a stable form.
Refinement of languages during concatenation-elimination helps to prune infeasible case splits that would otherwise continue branching only to be discarded at the very end of the decision procedure, at the generation of the length image.





\begin{example}[Stabilization vs. concatenation elimination]\label{example:stabilvsconcatelim}
	Let $\cube = \ft{T}(uvw, xyz) \land \constraintsRegulars \land \lenvar{z} = 8 + 3k$ be a cube.
	To derive a length image, techniques based on elimination of concatenation split the input $uvw$ and output $xyz$ terms to generate all possible alignments of variables such that a word from $\regLangConcat{uvw}$ can be transduced to a word from $\regLangConcat{xyz}$.
	This produces many cases to consider: e.g.,
	$\transductionArgsOf{T_1}{u_1}{x} \land \transductionArgsOf{T_2}{u_2}{y_1} \land \transductionArgsOf{T_3}{v_1}{y_2} \land \transductionArgsOf{T_4}{v_2}{z_1} \land \transductionArgsOf{T_5}{w}{z_2} \land \equationNodeArgs{u}{u_1u_2} \land \equationNodeArgs{v}{v_1v_2} \land \equationNodeArgs{y_1y_2}{y} \land \equationNodeArgs{z_1z_2}{z}$;
	or $\transductionArgsOf{T_1}{u}{x_1} \land \transductionArgsOf{T_2}{v}{x_2} \land \transductionArgsOf{T_3}{w_1}{x_3} \land \transductionArgsOf{T_4}{w_2}{y} \land \transductionArgsOf{T_5}{w_3}{z} \land \equationNodeArgs{w}{w_1w_2w_3} \land \equationNodeArgs{x_1x_2x_3}{x}$.
	Generally up to exponentially many to the number of repetitions of variables.
	Since the stable part $\stablecOf{\cube}$ in a stable-\concatenationFree form does not influence the length image of $\cube$, we can split $\cube$ into a stable part $\stablecOf{\cube} \colon \ft{T}(uvw, s) \land \equationNodeArgs{s}{s_1s_2s_3} \land \equationNodeArgs{s_1}{x}, \equationNodeArgs{s_2}{y}$ and a \concatenationFree part $\solvedcOf{\cube} \colon \equationNodeArgs{s_3}{z}$, $\sharedv = \{ s_3 \}$ such that
	$\stablecOf{\cube}$ can be stabilized without any splits (stability is enough for satisfiability), and the length image can be extracted solely from $\solvedcOf{\cube} \land \constraintsRegulars$.
	Hence, stabilization avoids the state-space blow-up of variable alignments, and prunes infeasible cases early.
	\sectionEndSymbol
\end{example}

\scaledvspace{-4mm}
\section{Noodlification for Equations and Transducer Constraints} \label{sec:noodlification}
\scaledvspace{-2mm}

The core of our decision procedure is the transformation of the input cube into the stable-\concatenationFree form by a sequence of \emph{noodlification} steps of two kinds, introduced for equations in~\cite{NoodlerFM23} and later extended to equations with length constraints in~\cite{NoodlerOOPSLA23}.
They refine languages of variables to achieve stability; one additionally splits equations and transducer constraints (and their variables) to achieve concatenation-freeness.
Noodlification steers the generation of possible equation/transducer splits (alignments) only to the feasible ones.
A~\emph{split} is a splitting of variables into a concatenation of fresh variables where boundaries between variables on both sides of an equation/transducer constraint match.
A~\emph{feasible split} is such where the fresh variables have non-empty languages.

We first recall the noodlification algorithms for the equations of~\cite{NoodlerFM23,NoodlerOOPSLA23}, and then describe both noodlification algorithms for transducers.
The algorithms for transducers are similar to those for equations.
The extension from equations to transducers is nevertheless nontrivial, especially in the case of noodlification that eliminates concatenation.

To fit our decision procedure in~\cref{sec:dec-proc},
all noodlification variants use a~unified interface transforming a directly chain-free cube $\formula\colon \constraintsRegulars \land \eta$
where $\eta$ is an equation or a~transducer constraint to a disjunction of directly chain-free cubes $\psi \colon \constraintsRegulars_\psi \land \aligneqin_\psi \land \aligneqout_\psi \land \constraintsTransducers_\psi$ where
\begin{inparaenum}[(i)]
	\item $\constraintsRegulars_\psi$ are (refined) regular constraints in a form closer to stability than $\constraintsRegulars$,
	\item $\aligneqin_\psi$ and $\aligneqout_\psi$ are \emph{input} and \emph{output binding equations} of the form $x \equals \sterm$ and $\sterm \equals x$, respectively, and
	\item $\constraintsTransducers_\psi$ are (refined) transducer constraints.
\end{inparaenum}
Note that the binding equations are not fully concatenation-free, since $x$ may occur in multiple equations, and the output binding equations admit $x$ on the right.
A variable split is performed together with splitting the variable's languages in~$\constraintsRegulars$.
If $\Regof{\psi}(x) \supseteq \Regof{\psi}(\sterm)$ for each input binding equation $x \equals \sterm$ and output binding equation $\sterm \equals x$ and $\Trof{\psi}$ is concatenation-free, we call $\psi$ \emph{almost \concatenationFree}.
The (input/output) binding equations in almost \concatenationFree cube are in the decision procedure of \cref{sec:dec-proc}, further processed, and turned into a fully \concatenationFree form.

\scaledvspace{-2mm}
\subsection{Noodlification for Equations}\label{sec:equationnoodlification}
\scaledvspace{-2mm}
Although equations could be in theory treated as a case of transducer constraints (the relation being identity),
the noodlification specialized to equations can be implemented more efficiently.
We therefore use the procedures specialized for equations from \cite{NoodlerOOPSLA23} whenever possible, also as a part of the workflow for processing transducer constraints.

\paragraph{Stabilizing noodlification.}
We invoke the procedure $\noodlifyrefine$ from \cite{NoodlerOOPSLA23} by calling $\eqtostable(\constraintsRegulars \land \equationNodeArgs{s}{t})$.
It returns a disjunction of cubes of the form $\psi \colon \constraintsRegulars_\psi \land \top \land \top \land \top$.
In conjunction with $\equationNodeArgs{s}{t}$, the disjunction is equivalent to the original cube, and for each disjunct, $\equationNodeArgs{s}{t}\land \constraintsRegulars_\psi$ is stable.
Stabilization treats an equation $\equationNodeArgs{s}{t}$ as a pair of inclusions $s \subseteq t$ and $t \subseteq s$ and refines only the variables on the left side of the processed inclusion.
We use NFAs for the languages of variables (from initial regular constraints) and perform only regularity-preserving operations.
For $s = x_1\cdots x_n$ and $t = y_1\cdots y_m$, we construct an $\epsilon$-preserving product NFA for the delimited language $w_1\markerL\cdots\markerL w_n$ with each $w_i \in \lass(x_i)$ and $w_1\cdots w_n \in \lass(t)$, treating $\markerL$ as $\epsilon$ but preserving it in the product.
This yields a tight set of refinements: every solution of $s=t$ under $\lass$ is captured by at least one refinement of $\noodlifyrefine$.
The product NFA is split into \emph{noodles}: sub-automata that contain exactly $n-1$ delimiter transitions (one delimiter transition per each \emph{level} of delimiter transitions in the product given by the borders between two variable languages in $s$).
Their union preserves the product language.
Each noodle is cut at its delimiter transitions into segments, one per occurrence of a variable in $s$.
For each variable $x_i$, the segments of its occurrences are intersected to produce the refined language of $x_i$.
Noodles that make some variable empty are discarded.

\paragraph{Concatenation-eliminating noodlification.}
We invoke the procedure $\noodlifysplit$ from \cite{NoodlerOOPSLA23} by calling $\eqtomonadic(\constraintsRegulars \land \equationNodeArgs{s}{t})$.
It returns a disjunction of almost \concatenationFree cubes of the form $\psi \colon \constraintsRegulars_\psi \land \aligneqin_\psi \land \aligneqout_\psi \land \top$ that is an existential extension of the input cube ($\aligneqin$ and $\aligneqout$ are in \cite{NoodlerOOPSLA23} unified in a set of alignment equations).
We make the alignment explicit for concatenation-eliminating noodlification by marking \emph{both} sides with distinct delimiters $\markerL$ and $\markerR$ and constructing the corresponding product automaton that preserves both kinds of delimiters.
Each noodle then represents a concrete interleaving of the left/right borders and induces a split into fresh variables together with alignment equations, producing a disjunction of almost \concatenationFree cubes that is an existential extension of the input.

\begin{example}{Concatenation-eliminating noodlification for equations.}\label{example:noodlification-equation}
	For a chain-free cube $\formula\colon \constraintsRegulars \land xyu = wz$ with $\constraintsRegulars = \{ x \mapsto a^*,\; p \mapsto (a+b)^* \mid p \in \{ y, u ,w ,z \} \}$, one feasible split $\psi$ generated by the concatenation-eliminating noodlification is $\psi \colon \constraintsRegulars_\psi \land \aligneqin_\psi \land \aligneqout_\psi \land \top$ where
	\begin{inparaenum}[(i)]
		\item $\constraintsRegulars_\psi = \constraintsRegulars \cup \{ v_1 \mapsto a^*,\; v_2 \mapsto a^*,\; v_3 \mapsto (a+b)^*,\; v_4 \mapsto (a+b)^* \}$,
		\item $\aligneqin_\psi = \{ x = v_1 v_2,\; y = v_3,\; u = v_4 \} $ and $\aligneqout_\psi = \{ v_1 = w ,\; v_2v_3v_4 = z \} $.
	\end{inparaenum}
	\sectionEndSymbol
\end{example}

\scaledvspace{-2mm}
\subsection{Concatenation-eliminating Noodlification for Transducers}
\label{sec:noodlemonadic}
\scaledvspace{-2mm}
Concatenation-eliminating noodlification transforms a cube $\formula = \constraintsRegulars \land \transductionArgs{s}{t}$, with $s = x_0 \cdots x_n$ and $t = y_0\cdots y_m$ of a transducer constraint and regular constraints into an almost concatenation-free and stable form.

\scaledvspace{-1mm}
\paragraph{Synchronization of transducer and regular constraints.}
In the first step, we construct a version of $\transduction$ with its input and output restricted by $\regLangConcat{s}$ and $\regLangConcat{t}$, respectively, and with delimited borders between variables in $s$ and $t$ (as in the product NFA for noodlification for equations).
Specifically, we construct the transducer
$
	\noodlificationNftLang = \syncOnTapesMarked
	{
		\big(
		\syncOnTapesMarked
		{\regLangConcatPresMark{s}{\markerL}}
		{\transduction}
		{1}{1}{\markers}
		\big)}
	{\regLangConcatPresMark{t}{\markerR}}
	{2}{1}
	{\markers}
$
where $\markers = \{\markerL, \markerR\}$ is the set containing \emph{input} ( $\markerL$ ) and \emph{output delimiter} ( $\markerR$ ), and for a delimiter $\marker \in \markers$, $\regLangConcatPresMark{z_1\concat z_2 \concats z_k}{\marker} = \langAssignOf{z_1} \concat \marker \concat \langAssignOf{z_2} \concat \marker \concats \marker \concat \regLangConcat{z_k}$.
The synchronization of $\markers$-delimited transducers guarantees that the delimiters never mix with other symbols in a transition.
The language of $\noodlificationNftLang$ thus consists of words
$
	w_{i_0,j_0} \concat \relationIdPair{\marker_1} \concat w_{i_1,j_1} \concat \relationIdPair{\marker_2} \concat \ldots \concat \relationIdPair{\marker_{n+m}} \concat w_{i_{n+m},j_{n+m}}
$
with each $\marker_k \in \markers$ marking the boundaries between the variables of $s$ (with input delimiter $\markerL$) and of $t$ (with output delimiter $\markerR$).
Intuitively, $L(T)$ would have a word $w \in (\Sigma_\epsilon^2)^*$ encoding a pair $(u,v)$
with $u = u_{x_1}\cdots u_{x_n} \in \langAssign(s)$ and $v = v_{y_1}\cdots {v_{y_m}} \in \langAssign(t)$.
The split of $w$ into the concatenation of segments $w_{i_0,j_0}\concat\ldots\concat w_{i_{n+m},j_{n+m}}$ with each $w_{i_k,j_k}$ encoding a pair of words $(\bar{u}_{i_k},\bar{v}_{j_k})$ encodes an \emph{alignment of variables} of $s$ and $t$ generated by methods such as~\cite{AnthonyTowards2016,ChainFree}.
The alignment may be seen as a choice of how to overlap variables of the two sides $s$ and $t$.
I.e., $\bar{u}_{i_0}\cdots \bar{u}_{i_{n+m}} = u$, $\bar{v}_{j_0}\cdots \bar{v}_{j_{n+m}} = v$,
each word segment $\bar{u}_{i_k}$ is a substring of $u_{x_{\ell}}$ for some input variable $x_\ell$, and
$\bar{v}_{j_k}$ a substring of $v_{y_\ell}$ for some output variable $y_\ell$.
Our automata representation of the alignments is efficient: it allows sharing and early pruning of alignments infeasible due to the incompatibility of regular constraints on variables.

%
For example, for $n=3$ and $m=2$,
one of the words with highlighted parts forming input and output variables could be
%
$$
	\mathrlap{\underbrace{\phantom{
				w_{0,0} \relationIdPair{\markerL} w_{1,0} \relationIdPair{\markerL} w_{2,0}
			}}_{y_0}}
	\mathrlap{\overbrace{\phantom{w_{0,0}}}^{x_0}}
	w_{0,0}
	\relationIdPair{\markerL}
	\mathrlap{\overbrace{\phantom{w_{1,0}}}^{x_1}}
	w_{1,0}
	\relationIdPair{\markerL}
	\mathrlap{\overbrace{\phantom{w_{2,0} \relationIdPair{\markerR} w_{2,1}}}^{x_2}}
	w_{2,0}
	\relationIdPair{\markerR}
	\mathrlap{\underbrace{\phantom{w_{2,1} \relationIdPair{\markerL} w_{3,1}}}_{y_1}}
	w_{2,1}
	\relationIdPair{\markerL}
	\mathrlap{\overbrace{\phantom{w_{3,1} \relationIdPair{\markerR} w_{3,2}}}^{x_3}}
	w_{3,1}
	\relationIdPair{\markerR}
	\mathrlap{\underbrace{\phantom{w_{3,2}}}_{y_2}}
	w_{3,2}\text{.}
$$
%
%
%
Importantly, $\noodlificationNftLang$ preserves all solutions of the original constraint: every solution of the constraint $\varphi$ corresponds to a run of $\noodlificationNftLang$ where the values of the individual variables are between the occurrences of delimiters ($\markerL$ on the input, $\markerR$ on the output).
The opposite (a run of $\noodlificationNftLang$ corresponds to a solution of the constraint $\varphi$) does not hold because the output of $T$ might contain multiple occurrences of the same variable (in a single run of $\noodlificationNftLang$ the words corresponding to the same variable are not synchronized).
The synchronization of different occurrences of the same variables is handled by the decision procedure later.


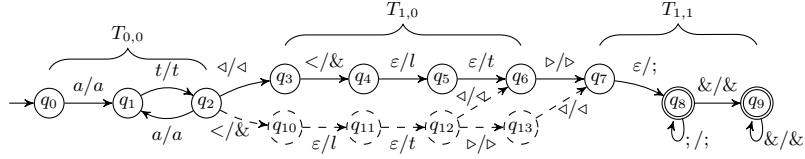
\begin{figure}[t]
	\centering
	\begin{tikzpicture}[>=stealth',shorten >=0pt,auto,node distance=13mm,transform shape,scale=0.8]
  \tikzset{every state/.style={minimum size=width("$q_{10}$")+2pt,inner sep=1pt,}}

  \node[state, initial, initial text={}] (q0) {$q_0$};
  \node[state, right of=q0] (q1) {$q_1$};
  \node[state, right of=q1] (q2) {$q_2$};
  \node[state, right of=q2,yshift=4mm] (q3) {$q_3$};
  \node[state, right of=q3] (q4) {$q_4$};
  \node[state, right of=q4] (q5) {$q_5$};
  \node[state, right of=q5] (q6) {$q_6$};
  \node[state, right of=q6] (q7) {$q_7$};
  \node[state, accepting, right of=q7,yshift=-4mm] (q8) {$q_8$};
  \node[state, accepting, right of=q8] (q9) {$q_9$};

  \draw[->] (q0) edge node{$a/a$} (q1);
  \draw[->] (q1) edge[bend left] node{$t/t$} (q2);
  \draw[->] (q2) edge[bend left] node{$a/a$} (q1);
  \draw[->] (q2) edge[bend left=10,pos=0.7] node{$\markerL/\markerL$} (q3);
  \draw[->] (q3) edge node{${<}/\&$} (q4);
  \draw[->] (q4) edge node{$\varepsilon/l$} (q5);
  \draw[->] (q5) edge node{$\varepsilon/t$} (q6);
  \draw[->] (q6) edge node{$\markerR/\markerR$} (q7);
  \draw[->] (q7) edge[bend left=10,pos=0.2] node{$\varepsilon/;$} (q8);
  \draw[->] (q8) edge[loop below] node[right]{$;/;$} (q8);
  \draw[->] (q8) edge node{$\&/\&$} (q9);
  \draw[->] (q9) edge[loop below] node[right]{$\&/\&$} (q9);

  \draw[decorate, decoration={brace, amplitude=8pt}] 
    ([yshift=3mm]q0.north) -- ([yshift=3mm]q2.north) 
    node[midway, above=9pt]{$T_{0,0}$};

  \draw[decorate, decoration={brace, amplitude=8pt}] 
    ([yshift=3mm]q3.north) -- ([yshift=3mm]q6.north) 
    node[midway, above=9pt]{$T_{1,0}$};

  \draw[decorate, decoration={brace, amplitude=8pt}] 
    ([yshift=3mm]q7.north) -- ([yshift=7mm]q9.north) 
    node[midway, above=9pt]{$T_{1,1}$};

  \begin{scope}[dashed]

  \node[state, right of=q2, yshift=-4mm] (q10) {$q_{10}$};
  \node[state, right of=q10] (q11) {$q_{11}$};
  \node[state, right of=q11] (q12) {$q_{12}$};
  \node[state, right of=q12] (q13) {$q_{13}$};
  \draw[->] (q2) edge[bend right=10] node[below,pos=0.2]{${<}/\&$} (q10);
  \draw[->] (q10) edge node[below]{$\varepsilon/l$} (q11);
  \draw[->] (q11) edge node[below]{$\varepsilon/t$} (q12);
  \draw[->] (q12) edge node[below]{$\markerR/\markerR$} (q13);
  \draw[->] (q12) edge node[left,pos=0.7]{$\markerL/\markerL$} (q6);
  \draw[->] (q13) edge node[right,pos=0.3]{$\markerL/\markerL$} (q7);
  \end{scope}
\end{tikzpicture}
	\vspace{-0.2cm}
	\caption{Compositional transducer $\noodlificationNftLang$ constructed during concatenation-eliminating noodlification for constraint in Example~\ref{ex:stupidex}. The non-dashed part is one of its noodles.}
	\label{fig:running-example/compositional-nft}
	\vspace{-4mm}
\end{figure}

\scaledvspace{-1mm}
\paragraph{Case-split by noodles.}
In the second step, $\noodlificationNftLang$ is split into a finite set $\mathit{Noodles}$ of transducers $\noodle$ (called \emph{noodle}) of the form
$
	\ftT_{i_0,j_0} \concat \relationIdPair{\marker_1} \concat \ftT_{i_1,j_1} \concat \relationIdPair{\marker_2} \concat \ldots \concat \relationIdPair{\marker_{n+m}} \concat \ftT_{i_{n+m},j_{n+m}}
$
where each $T_{i_k,j_k}$ is a transducer over $\SigmaEps$ and $\lang(\noodlificationNftLang) = \bigcup_{\noodle \in \mathit{Noodles}} \lang(\noodle)$.
Intuitively, each $\noodle$ is a long, thin transducer formed as a sequence of component \emph{segment transducers} $\ft{T}_{i_k,j_k}$ connected by delimiters, resembling strands or ``noodles''.
For each \emph{level} of delimiters, we use only one of the transitions in $\noodlificationNftLang$.
%
The indices $i_k$ and $j_k$ count the number of preceding input $(\markerL,\markerL)$ and output $(\markerR,\markerR)$ pairs.

Finally, we create a conjunction $\psi_\noodle = \constraintsRegulars \land \aligneqin \land \aligneqout \land \constraintsTransducers$ for each noodle $\noodle$ where
\begin{enumerate*}
	\item $\constraintsTransducers_\noodle$ are \emph{segment (transducer) constraints} $T_{i_k,j_k}(z_{i_k,j_k},\bar{z}_{i_k,j_k})$,
	\item $\constraintsRegulars_\noodle$ are regular constraints $\constraintsRegulars$ extended with \emph{new regular constraints} of the form $z_{i_k,j_k} \in \projTo{1}{\ftT_{i_k,j_k}}$ and $\bar{z}_{i_k,j_k} \in \projTo{2}{\ftT_{i_k,j_k}}$ where $z_{i_k,j_k},\bar{z}_{i_k,j_k}$ are fresh variables,
	\item $\aligneqin_\noodle$ are \emph{input binding equations} of the form $x_\ell \approx z_{\ell,g_1} \cdots z_{\ell,g_h}$ where $g_1 \cdots g_h$ covers the interval of the segment transducers between two $(\markerR,\markerR)$ for the corresponding index $\ell$, and
	\item $\aligneqout_\noodle$ are \emph{output binding equations} of the form $\bar{z}_{g_1,\ell} \cdots \bar{z}_{g_h,\ell} \approx y_\ell$ where $g_1 \cdots g_h$ covers the interval of the segment transducers between two $(\markerL,\markerL)$ for the corresponding index $\ell$.
\end{enumerate*}

%
Thus, the concatenation-eliminating noodlification returns for the cube $\formula = \transductionArgs{s}{t} \land \constraintsRegulars$ the disjunction of almost \concatenationFree cubes $\noodlifyMonadicFt(\formula) = \bigvee_{\noodle \in \mathit{Noodles}} \psi_\noodle$.

%
%
%

\scaledvspace{-1mm}
\begin{restatable}{lemma}{lemmaMonadicExistentialExtension}\label{lemma:monadic-existential-extension}
	$\noodlifyMonadicFt(\formula)$ is a disjunction of almost \concatenationFree cubes, and it is an existential extension of $\formula$.
\end{restatable}
\scaledvspace{-2mm}

\begin{wrapfigure}[12]{r}{0.25\textwidth}
	\vspace{-1.2cm}
	\centering
	$$
		\begin{array}{c|c|c}
			\texttt{[alt]}^+\texttt{<}^*                          & \multicolumn{2}{c}{\texttt{<}^*\texttt{[;\&]}^*}                             \\
			x                                                     & \multicolumn{2}{c}{y}                                                        \\
			\hline
			\texttt{(at)}^+                                       & \texttt{<}                                       & \texttt{;}^*\texttt{\&}^* \\
			z_{0,0}                                               & z_{1,0}                                          & z_{1,1}                   \\
			T_{0,0}                                               & T_{1,0}                                          & T_{1,1}                   \\
			\bar{z}_{0,0}                                         & \bar{z}_{1,0}                                    & \bar{z}_{1,1}             \\
			\texttt{(at)}^+                                       & \texttt{\&lt}                                    & \texttt{;}^+\texttt{\&}^* \\
			\hline
			\multicolumn{2}{c|}{v}                                & z                                                                            \\
			\multicolumn{2}{c|}{\texttt{(at)}^*\texttt{(\&lt)}^+} & \texttt{;}^*\texttt{\&}^*                                                    \\
		\end{array}
	$$
	\vspace{-0.4cm}
	\caption{Elements of $\psi$.}
	\label{fig:running-example/possible-alignment}
\end{wrapfigure}
\medskip\noindent\refstepcounter{example}\textit{Example~\theexample}\label{ex:stupidex}.
Let $\Sigma = \{ a, <, \&, l, t, ; \}$ be an alphabet, $\formula = \constraintsRegulars \land \relConstraintArgs{T}{xy}{vz}$ a cube where $\ft{T}$ is the transducer from~\cref{example:modelling-replace-operations}, replacing all occurrences of $<$ with $\texttt{\&lt;}$.
We use the standard regular expressions
to define regular languages.
$\constraintsRegulars$ is given as $\langAssignOf{x} = \texttt{[alt]}^+\texttt{<}^*$, $\langAssignOf{y} = \texttt{<}^*\texttt{[;\&]}^*$, $\langAssignOf{v} = \texttt{(at)}^*\texttt{(\&lt)}^+$, and $\langAssignOf{z} = \texttt{;}^*\texttt{\&}^*$.
%
\cref{fig:running-example/compositional-nft} shows a~compositional transducer $\noodlificationNftLang$ with the non-dashed part being 
one of its noodles $\noodle$ with its segment transducers $T_{i_k, j_k}$ (the other noodles follow the lower paths). 
\cref{fig:running-example/possible-alignment} schematically represents elements of $\psi_\noodle$.
The vertical lines show the boundaries between the variables in $s$ and $t$ ($\markerL$ for vertical lines that split the input row,  $\markerR$ for the output row).
The inner rows show the fresh variables $z_{i_k,j_k}$ and $\bar{z}_{i_k,j_k}$, their correspondence to the segment transducer constraints $T_{i_k,j_k}(z_{i_k,j_k},\bar{z}_{i_k,j_k})$, and their refined languages.
The outermost rows show the original variables with their languages.
$\noodle$ leads to the following elements of $\psi_\noodle$:
\begin{inparaenum}[(i)]
	\item $\constraintsRegulars_\noodle \colon \constraintsRegulars \land
		z_{0,0} \in \texttt{(at)}^+ \land z_{1,0} \in \texttt{<} \land z_{1,1} \in \texttt{;}^*\texttt{\&}^* \land
		\bar{z}_{0,0} \in \texttt{(at)}^+ \land \bar{z}_{1,0} \in \texttt{\&lt} \land \bar{z}_{1,1} \in \texttt{;}^+\texttt{\&}^*
	$,
	\item $\aligneqin_\noodle \colon \equationNodeArgs{x}{z_{0,0}} \land \equationNodeArgs{y}{z_{1,0}z_{1,1}}$,
	\item
	$\aligneqout_\noodle \colon \equationNodeArgs{\bar{z}_{0,0}\bar{z}_{1,0}}{v} \land \equationNodeArgs{\bar{z}_{1,1}}{z}$, and
	\item $\constraintsTransducers_\noodle \colon \relConstraintArgs{\ft{T_{0,0}}}{z_{0,0}}{\bar{z}_{0,0}} \land \relConstraintArgs{\ft{T_{1,0}}}{z_{1,0}}{\bar{z}_{1,0}} \land \relConstraintArgs{\ft{T_{1,1}}}{z_{1,1}}{\bar{z}_{1,1}}$.
\end{inparaenum}
\sectionEndSymbol
\scaledvspace{-2mm}
\subsection{Stabilizing Noodlification for Transducers}
\label{sec:noodlestable}
\scaledvspace{-2mm}
In this section, we describe how to transform a conjunction $\formula = \constraintsRegulars \land \transductionArgs{s}{t}$ of a transducer constraint with normalized regular constraints into the stable form, without trying to eliminate concatenation.
We can simplify our task of achieving the stability of $\transductionArgs{s}{t}$ by rewriting it as an equivalent constraint $\transductionArgs{s}{z} \land \equationNodeArgs{z}{t}$ where $z$ is a fresh string variable.
The transducer constraint $\transductionArgs{s}{z}$ is simpler, as it contains only a single variable in the output, and the equation $\equationNodeArgs{z}{t}$ can be stabilized by means of~\cite{NoodlerFM23}.
The $\transductionArgs{s}{z}$ is stabilized simply by setting the language of $z$ to $\relationApplyOn{\transduction}{\regLangConcat{s}}$.

The stable noodlification returns for the cube $\formula$ the formula $\noodlifyStableFt(\formula) = \constraintsRegulars_s \land \aligneqin_s \land \top \land \constraintsTransducers_s$ where
$\constraintsTransducers_s = \transductionArgs{s}{z}$,
$\aligneqin_s = \equationNodeArgs{z}{t}$,
and
$\constraintsRegulars_s = (\constraintsRegulars \land z\in\relationApplyOn{\transduction}{\regLangConcat{s}})$.

Note that formula $\noodlifyStableFt(\formula)$ is not necessary in stable form, as
condition $\regLangConcat{z} \supseteq \regLangConcat{t}$ does not need to hold for $\equationNodeArgs{z}{t}$.
In the decision procedure of \cref{sec:dec-proc}, the equation will eventually stabilize via noodlification for equations.

\scaledvspace{-2mm}
\paragraph{Stabilizing vs. concatenation-eliminating noodlification.}
Transformation of $T(s,t)$ to $T(t,z) \land z\equals t$ followed by the stabilizing noodlification of $z\equals t$ only propagates refinements to the languages of the variables of $t$.
Instead of generating noodles from $T_+$ as in concatenation-eliminating noodlification in \cref{sec:noodlemonadic},
noodles are generated from a simpler product without explicit marked borders of the input variables: $\syncOnTapesMarked{\relationApplyOn{\transduction}{\regLangConcat{s}}}{\regLangConcatPresMark{t}{\markerR}}{1}{1}{\{\markerR\}}$.
This produces far fewer noodles than the
\concatenationFree noodlification (we avoid generating interleavings of input and output variables)
and does not require introducing new variables.
The effect compared to the concatenation-eliminating noodlification is up to exponentially fewer noodles in a single noodlification step and exponentially fewer new variables across multiple noodlification steps.

\scaledvspace{-1mm}
\section{Decision Procedure}\label{sec:dec-proc}
\scaledvspace{-2mm}

\begin{wrapfigure}[17]{r}{0.55\textwidth}
	\vspace{-10mm}
	\begin{algorithm}[H]
		\footnotesize
		\caption{Decision procedure}
		\label{alg:satisfiability}

		\KwIn{A directly chain-free cube $\formula$}
		\KwOut{$\sat$ if $\formula$ is satisfiable; \newline $\unsat$ if $\formula$ is unsatisfiable}

		$\rules := [\ruleLSubs{}$, $\ruleLLSubs{}$, $\ruleSkip{}$, $\refineStable{}$, $\ruleCombineNonLength{}$, $\ruleAlignSplitLen{}]$\;
		$W := \{(\EqTrof\cube$, $\texttt{toposort}(\EqTrof\cube)$, \hspace*{0.83cm}$\Regof\cube$, $\emptyset$, $\varsOf{\Lenof\cube})\}$\;
		\While{$W$ is not empty}{
			$V := W.\mathit{pop()}$\; 
			\If{$\ruleStable$ is applicable for $V$}{
				\Return $\sat$;
			}
			\For{$\mathit{rule} \in \rules$}{
				\If{$\mathit{rule}$ is applicable for $V$}{
					$W = W \cup \mathit{rule}(V)$\;
					\Break\;
				}
			}
		}
		\Return $\unsat$\;
	\end{algorithm}
\end{wrapfigure}

In this section, we describe the decision procedure for directly chain-free cubes.
The decision procedure first transforms the input cube into a stable-\concatenationFree form and then constructs and checks the satisfiability of the length image of the \concatenationFree part.
The division of the formula into the \concatenationFree and the stable part is derived on the fly, driven by accumulating the set of \emph{length-aware} variables---with values-dependent length const\-raints.
%
Constraints with length-aware variables in the input end up in the concatenation-free part of the final stable-\concatenationFree formula.

The procedure builds a proof tree by repeatedly applying the \emph{inference rules} presented later in \cref{sec:refrules,sec:substrules}.
The vertices of the proof tree are tuples $(\incls$, $\frontier$, $\lass$, $\solvedeqs$, $\lenvars)$ where
$\incls$ is a set of equations and transducer constraints of the inclusion graph of the formula, $\frontier \subseteq \incls$ is an inclusion graph exploration frontier of constraints about to be processed, with elements always maintained ordered in the topological order, $\lass$ is a language assignment, $\solvedeqs$ is a set of solved equations, and $\lenvars$ is a set of length-aware variables.
A vertex represents the cube $\incls \land \lass \land \solvedeqs$. 
%

We write the inference rules as
$\infer[\ruleName{Name}]{\ikset{(\incls_i, \frontier_i, \lass_i, \solvedeqs_i, \lenvars_i)}}{(\incls, \frontier, \lass, \solvedeqs, \lenvars)}{~\varphi_{\ruleName{Name}}}$
where $\ruleName{Name}$ is the rule's \emph{name},
the vertex above the line is the \emph{premise},
$\varphi_{\ruleName{Name}}$~is the~\emph{side condition} on the premise necessary for the rule to be applicable,
and the set of vertices under the line is the~\emph{set of conclusions}.
Each application consumes a premise and produces a set of conclusions, each conclusion spawning a new branch of the proof tree.
A set of conclusions encodes a disjunction as an existential extension of the premise vertex cube.

Each application of an inference rule processes a constraint on top of the frontier~$\frontier$, potentially refining the language assignment and splitting complex transducer constraints and equations that require decomposition (have length-aware variables in the input), while accumulating length-aware variables.
A constraint is processed only after all its predecessors in the inclusion graph were processed, hence the whole algorithm is based on a graph exploration.
All inference rules maintain the direct chain-freeness invariant of the cubes.
Constraints requiring decomposition are sometimes reinserted to the graph exploration frontier, pushing the frontier back.
In the following, we write $\ell\lconcat\frontier$ to denote the frontier where $\ell$ is the first element and $\frontier$ is the rest, and we write $M\lconcat\frontier$ to denote the frontier with the elements of a set $M$ put on the top (in topological order).
%

The pseudocode of the decision procedure is shown in \cref{alg:satisfiability}.
%
The list \rules{} has the inference rules (see \cref{sec:refrules,sec:substrules}) ordered by their priority.
The worklist $W$ contains yet to be processed proof tree vertices.
It is initialized with the vertex $(\EqTrof\cube$, $\texttt{toposort}(\EqTrof\cube)$, $\Regof\cube$, $\emptyset$, $\varsOf{\Lenof\cube})$ where $\texttt{toposort}(\EqTrof\cube)$ is the set of vertices of inclusion graph for $\varphi$ ordered in the topological order, $\Regof\cube$ is the initial regular constraint and the set of length-aware variables is initialized as $\varsOf{\Lenof\cube}$.
%
The procedure iteratively processes vertices $V$ from $W$.
If the rule $\ruleStable$ (see \cref{sec:lengths}) is applicable we have stable-\concatenationFree form and the length image is satisfiable, therefore according to \cref{lemma:stable-solved} we return $\sat$.
Otherwise, we proceed by selecting the first applicable inference rule for $V$, and adding the set of conclusions $\mathit{rule}(V)$ of the rule applied on $V$ to the worklist $W$.
If no proof tree vertex is left to process ($W$ is empty), it means that the original cube does not have a solution, and we return $\unsat$.

For the sake of compactness,
we will further write the input and output cubes of the noodlification procedures from \cref{sec:noodlification} as tuples of the sets of constraints: pairs $(\constraintsRegulars,\eta)$ for inputs and quadruples $(\constraintsRegulars,\aligneqin,\aligneqout,\constraintsTransducers)$ for outputs (or sets of quadruples for disjunctions of cubes). The $\emptyset$ here stands for the conjunct $\top$.

\scaledvspace{-1mm}
\subsection{Refinement Rules}\label{sec:refrules}
\scaledvspace{-1mm}
%
We start with refinement rules. These rules transform an equality or transducer constraint at the top of $\frontier$ into constraints closer to stable or concatenation-free form, ultimately producing a stable-concatenation-free form.
Choosing between targeting the stable or the concatenation-free form is steered by the presence of length-aware variables: if the input of the constraint contains no length-aware variables, we refine toward the more efficient stable form; otherwise, we target the concatenation-free form.

The first refinement rule is $\refineStable{}$.
It transforms an equality/transducer constraint into constraints that are closer to stable form.
The transformation uses stabilizing noodlification (see \cref{sec:equationnoodlification,sec:noodlestable}) and may produce multiple successor
vertices, one for each tuple returned by the noodlification.
Recall that for an equation constraint, noodlification returns a set of tuples containing only refinements of $\lass$ setting up stable form of the equation.
For a transducer constraint, noodlification returns a single
tuple containing refined $\lass$, a transducer constraint and the equation, which should be returned to the frontier for further refinement. Although the refined transducer constraint itself is already in stable form, we keep it in $\incls$ for potential further refinement caused by substitutions (see \cref{sec:substrules}). Formally, the rule $\refineStable{}$ is given as:
%
%
{
\small
\[
	\refineStable{}:
	\setlength{\arraycolsep}{4.5pt}
	\begin{array}{rcccccl}
		(        & \shrinkit\incls \uplus \{\eta(\sterm, \tterm)\},                        & \eta(\sterm, \tterm)\lconcat \frontier, & \lass,   & \solvedeqs, & \lenvars & \shrinkit)
		\\
		\hline
		\ikbigl( & \shrinkit \incls \cup \aligneqin_i \cup \simpleTransducerConstraints_i, & \aligneqin_i\lconcat \frontier,         & \lass_i, & \solvedeqs, & \lenvars & \shrinkit)\ikbigr
	\end{array}\
	\varphi_{\refineStable{}}
\]
\vspace*{-5mm}
\begin{align*}
	\textstyle
	\text{where }
	 & \varphi_{\refineStable{}} \defiff
	\
	\varsOf{\sterm} \cap \lenvars = \emptyset\text{ and } \noodlifyStable(\eta(\sterm, \tterm) ,\lass) = \big\{(\lass_i, \aligneqin_i,\emptyset, \simpleTransducerConstraints_i)\big\}_{i=1}^k\text{.}
\end{align*}
}%
Here $\noodlifyStable(\eta(\sterm, \tterm), \lass)$ calls $\eqtostable$ if $\eta$
is an equation (in that case both $\aligneqin_i$ and $\simpleTransducerConstraints_i$ are empty), or it calls $\noodlifyStableFt$ if $\eta$ is a transducer (in that case $k=1$).

\begin{example}[Applying $\refineStable{}$ to~\cref{ex:stupidex}]\label{ex:refinestable}
	Consider the cube $\formula$ from~\cref{ex:stupidex} with $\lenvars = \emptyset$ (no length constraints).
	The initial proof tree vertex is $V_0 = (\{\relConstraintArgs{T}{xy}{vz}\},\; [\relConstraintArgs{T}{xy}{vz}],\; \lass_0,\; \emptyset,\; \emptyset)$,
	where $\lass_0$ maps $x \mapsto \texttt{[alt]}^+\texttt{<}^*$, $y \mapsto \texttt{<}^*\texttt{[;\&]}^*$, $v \mapsto \texttt{(at)}^*\texttt{(\&lt)}^+$, $z \mapsto \texttt{;}^*\texttt{\&}^*$.
	Since $\varsOf{xy} \cap \lenvars = \emptyset$, the rule $\refineStable{}$ applies to $\relConstraintArgs{T}{xy}{vz}$.
	It calls $\noodlifyStableFt(\relConstraintArgs{T}{xy}{vz},\; \lass_0)$, introducing a fresh variable $w$ and returning the single quadruple $(\lass_1,\; \{\equationNodeArgs{w}{vz}\},\; \emptyset,\; \{\relConstraintArgs{T}{xy}{w}\})$.
	The refined language assignment sets $\lass_1(w) = \relationApplyOn{T}{\regLangConcat{xy}}$.
	Since $\regLangConcat{xy} = \texttt{[alt]}^+\texttt{<}^*\texttt{[;\&]}^*$ and $T$ replaces every~$\texttt{<}$ by~$\texttt{\&lt;}$, $\lass_1(w) = \texttt{[alt]}^+\texttt{(\&lt;)}^*\texttt{[;\&]}^*$.
	The rule yields a single conclusion vertex
	$V_1 = (\{\relConstraintArgs{T}{xy}{w}, \\\equationNodeArgs{w}{vz}\}, [\equationNodeArgs{w}{vz}], \lass_1, \emptyset, \emptyset)$.
	\sectionEndSymbol
\end{example}

As a special case, to avoid refining constraints that are already stable, we use the rule $\ruleSkip{}$, allowing to skip stable constraints with no length variables in the input:
{
\small
\[\ruleSkip{}:
	\begin{array}{rcccccl}
		( & \incls, & \eta(\sterm, \tterm) \lconcat \frontier, & \lass, & \solvedeqs, & \lenvars & )
		\\
		\hline
		( & \incls, & \frontier,                               & \lass, & \solvedeqs, & \lenvars & )
	\end{array}\
	\varsOf{\sterm}\cap\lenvars=\emptyset \land
	(\eta(\sterm, \tterm) \land \lass) \text{ is stable}\text{.}
\]
}%

If the constraint on top of $\frontier$ contains a length variable in the input, we use the refinement rule $\ruleAlignSplitLen{}$ to employ the concatenation-eliminating noodlification (returning
a set of tuples containing refined regular constraints, input and output
binding equations, and in the case of transducer constraints, the concatenation-free segment transducer constraints).
As the segment transducer constraints are already concatenation-free, we do not
need to put them into $\frontier$, only into $\incls$ (as for the rule $\refineStable{}$).
We put in $\frontier$ only input and output binding equations since they might not be in the form of solved equations.
Since the noodlification may introduce fresh variables, we need to update the set
of length-aware variables so that information about which new variables affect the length constraint is propagated correctly.
Formally, the rule $\ruleAlignSplitLen{}$ is defined as:
%
{
\small
\[
	\ruleAlignSplitLen{}:
	\setlength{\arraycolsep}{2.5pt}
	\begin{array}{rcccccl}
		(        & \shrinkit\incls \uplus \{\eta(\sterm, \tterm)\},                         & \eta(\sterm, \tterm)\lconcat \frontier, & \lass,   & \solvedeqs, & \lenvars                 & \shrinkit)
		\\
		\hline
		\ikbigl( & \shrinkit \incls \cup \mathcal{E}_i \cup \simpleTransducerConstraints_i, & \mathcal{E}_i\lconcat \frontier,        & \lass_i, & \solvedeqs, & \lenvars \cup \lenvars_i & \shrinkit)\ikbigr
	\end{array}\
	\varphi_{\ruleAlignSplitLen{}}
\]
\vspace*{-4mm}
\begin{align*}
	\textstyle
	\text{where }
	 & \varphi_{\ruleAlignSplitLen{}} \defiff \varsOf{\sterm} \cap \lenvars \neq \emptyset \text{, } \noodlifyMonadic(\eta(\sterm, \tterm) ,\lass) = \big\{(\lass_i,\overbrace{\aligneqin_i, \aligneqout_i}^{\mathcal{E}_i = \aligneqin_i \cup \aligneqout_i},\simpleTransducerConstraints_i)\big\}_{i=1}^k\text{,} \\
	 & \hspace*{0cm} {}
	\text{and for } i \in \{1..k\}: \lenvars_i = \proplenfw(\simpleTransducerConstraints_i, \proplenfw(\aligneqin_i, \lenvars)) \cup \proplenbw(\aligneqout_i, \lenvars)\text{.}
\end{align*}
}%
Here, $\noodlifyMonadic(\eta(\sterm, \tterm) ,\lass)$ calls $\eqtomonadic$ if $\eta$
is an equation (in that case $\simpleTransducerConstraints_i$ is empty), or
$\noodlifyMonadicFt$ if $\eta$ is a transducer.
For a particular quadruple returned by the concatenation-eliminating noodlification, the set of length-aware variables is updated as follows.
The variables connected by $\aligneqin_i$ with the length-aware variables in the input of $\eta$ become length-aware.
Length awareness is propagated similarly from input to output via concatenation-free transducers $\simpleTransducerConstraints_i$.
We say these two are \emph{forward length propagations}, formally $\proplenfw(\col, \lenvars') = \lenvars' \cup \{ x \in \varsOf{t'} \mid \tau(s',t')\in\col, \varsOf{s'} \cap \lenvars' \neq \emptyset \}$.
The length-awareness of the variables in the output of $\eta$ is also propagated to the output variables of segment transducers through the output equalities $\aligneqout_i$.
This propagation is called \emph{backward length propagation} formally defined as $\proplenbw(\col, \lenvars') = \lenvars' \cup \{ x \in \varsOf{s'} \mid \tau(s',t')\in\col, \varsOf{t'} \cap \lenvars' \neq \emptyset \}$.

To reduce the number of unnecessary variable alignments and thus the number of noodles, the concatenation-free eliminating noodlification assumes that
at least one variable is length-aware of each two adjacent variables in the constraint's input.
Therefore, we use the rule $\ruleCombineNonLength{}$ to replace sequences of non-length-aware variables in the input with fresh variables, and generate new equations relating the fresh variables with replaced sequences.
As the original variables have no other occurrence in the constraint (by chain-freeness, since they appear on the input), we can exclude them from further processing, provided we assign the language of the concatenation of the original languages to the fresh variables.
Formally, the rule $\ruleCombineNonLength{}$ is given as:
%
{
\small
\[
	\ruleCombineNonLength{}:
	\setlength{\arraycolsep}{3pt}
	\begin{array}{rcccccl}
		( & \shrinkit\incls \uplus \{ \eta(\sterm, \tterm) \},                  & \eta(\sterm, \tterm)\lconcat \frontier,  & \lass,             & \solvedeqs, & \lenvars & \shrinkit)
		\\
		\hline
		( & \shrinkit \incls \cup \{ \eta(\sterm', \tterm) \} \cup \mathcal{E}, & \eta(\sterm', \tterm)\lconcat \frontier, & \lass' \cup \lass, & \solvedeqs, & \lenvars & \shrinkit)
	\end{array}\
	\varphi_{\ruleCombineNonLength{}}
\]
\vspace*{-5mm}
\begin{align*}
	\textstyle
	\text{where }
	\varphi_{\ruleCombineNonLength{}} & \defiff
	(\sterm', \mathcal{E}) = \uniform(\sterm, \lenvars) \text{, } \lass' = \{ x \mapsto \lass(u) \mid \equationNodeArgs{u}{x} \in \mathcal{E}  \}, \mathcal{E} \neq \emptyset \text{.}
\end{align*}
}%
Here, we define $\uniform(\sterm, \lenvars)$ as a pair $(\sterm', \mathcal{E})$ such that $\sterm'$ is formed from $\sterm$ where
each maximal contiguous block of non-length variables $x_1\cdots x_n$ where $n \geq 2$ is replaced by a fresh variable $x$
and the equation $\equationNodeArgs{x_1\cdots x_n}{x}$ is added to $\mathcal{E}$.

\scaledvspace{-1mm}
\subsection{Substitution Rules}\label{sec:substrules}
\scaledvspace{-1mm}
%
%
The substitution rules $\ruleLSubs{}$ and $\ruleLLSubs{}$ transform an output/input binding equation $\sterm \equals x$ or $x \equals \sterm$, respectively, into a proper solved equation (moving it from the almost concatenation-free to the fully concatenation-free part).
The transformation is constructed by applying the \emph{substitution} $\subst = \{x \mapsto \sterm\}$ induced by the equation on all constraints.
Application of the substitution $\subst$ on a constraint $\eta$ (equation or transducer), denoted $\subst(\eta)$, returns the constraint where each occurrence of $x$ is replaced with the term~$\sterm$.
We extend application of $\subst$ to sets in the natural way.

Applying a substitution may have two effects:
\begin{inparaenum}[(i)]
	\item some constraints from $\incls$ might start violating the stable-form condition, and
	\item some concatenation-free transducer constraints from $\incls$ might become non-concatenation-free.
\end{inparaenum}
Those constraints are returned to the frontier to be processed again.

The rule $\ruleLSubs{}$ transforms output binding equation $\sterm \equals x$ as follows:
{
\small
\[
	\ruleLSubs{}:
	\begin{array}{rcccccl}
		( & \incls \uplus \{ \equationNodeArgs{\sterm}{x} \},           & \equationNodeArgs{\sterm}{x} \lconcat \frontier, & \lass, & \solvedeqs, & \lenvars & )
		\\
		\hline
		( & \subst(\incls),                                             & H\lconcat \subst(\frontier),
		  & \lass,
		  & \subst(\solvedeqs) \cup \{ \equationNodeArgs{x}{\sterm} \}, & \lenvars
		  & )
	\end{array}\
	\varphi_{\ruleLSubs{}}
\]
\vspace*{-5mm}
\begin{align*}
	\textstyle
	\text{where }
	\varphi_{\ruleLSubs{}} & \defiff \subst = \{x \mapsto \sterm\} \text{, } (x\in\lenvars \limpl \varsOf{\sterm} \subseteq \lenvars) {, } \lass(\sterm) \subseteq \lass(x)\text{, and}           \\
	                       & H = \subst(\{\eta(\tterm, \tterm') \in\incls \mid \tterm \text{ contains } x\} \cup \{T(y,x) \mid y \in \lenvars, |\sterm| > 1\})\setminus \subst(\frontier)\text{.}
\end{align*}
}%
%
Intuitively, the rule takes an equation of the form $\equationNodeArgs{\sterm}{x}$ and applies the substitution $\subst$ on $\incls$, $\frontier$, and $\solvedeqs$.
To be applicable, $\sterm \equals x$ must be almost concatenation-free ($\lass(\sterm) \subseteq \lass(x)$) and furthermore, if $x$ is length-aware, all variables in $s$ must be length-aware (we cannot replace length-aware variable by non-length-aware ones).
In $H$, we have all the constraints that need to be put in frontier.
Because the language of $x$ was effectively refined (as $x$ was replaced with $\sterm$ whose language is smaller), we need to put all constraints $\eta$ with $x$ on input to frontier, as their stability might have been violated.
Furthermore, the transducer constraints $T(y,x)$ which were already concatenation-free (and need to be concatenation-free, i.e, the input $y$ is length-aware), need to be also put back in frontier.
Note that $x$ could be replaced by only one variable ($|\sterm| = 1$), in that case, the transducer constraint will stay concatenation-free, and we do not need to put it in frontier.
Finally, we add (the substituted) nodes to the frontier only if they are not in the frontier already.

\begin{example}[Applying $\ruleLSubs{}$]\label{ex:lsubs}
	Let $w$, $z$, $u$ be string variables,
	$\lass(w) = \texttt{;}^+\texttt{\&}^*$,
	$\lass(z) = \texttt{;}^*\texttt{\&}^*$,
	and $\relConstraintArgs{S}{z}{u}$ a
	transducer constraint with $z$ on its input.
	Consider the vertex
	$
		V = \bigl(
		\{\equationNodeArgs{w}{z},\; \relConstraintArgs{S}{z}{u}\},\;
		[\equationNodeArgs{w}{z}],\;
		\lass,\;
		\emptyset,\;
		\emptyset
		\bigr).
	$
	The equation $\equationNodeArgs{w}{z}$ is in output binding form,
	so $\ruleLSubs{}$ applies with $\subst = \{z \mapsto w\}$.
	The side condition $\varphi_{\ruleLSubs{}}$ holds:
	$\lass(w) = \texttt{;}^+\texttt{\&}^* \subseteq \texttt{;}^*\texttt{\&}^* = \lass(z)$
	and $z \notin \lenvars = \emptyset$.
	Since $\relConstraintArgs{S}{z}{u}$ has $z$ on its input, it enters $H$:
	$
		H = \subst\bigl(\{\relConstraintArgs{S}{z}{u}\}\bigr) \setminus \subst([\,])
		= \{\relConstraintArgs{S}{w}{u}\}.
	$
	The language of $z$ is effectively refined from $\texttt{;}^*\texttt{\&}^*$ to $\lass(w) = \texttt{;}^+\texttt{\&}^*$,
	which may violate stability of $\relConstraintArgs{S}{z}{u}$, so $\relConstraintArgs{S}{w}{u}$ is
	returned to the frontier.
	The conclusion vertex is:
	$
		V' = \bigl(
		\{\relConstraintArgs{S}{w}{u}\},\;
		[\relConstraintArgs{S}{w}{u}],\;
		\lass,\;
		\{\equationNodeArgs{z}{w}\},\;
		\emptyset
		\bigr).
	$
	\sectionEndSymbol
\end{example}

The rule $\ruleLLSubs$ for input binding equations is simpler:
{
\small
\[
	\ruleLLSubs:
	\begin{array}{rcccccl}
		( & \incls \uplus \{ \equationNodeArgs{x}{\sterm} \},           & \equationNodeArgs{x}{\sterm} \lconcat \frontier, & \lass, & \solvedeqs, & \lenvars & )
		\\
		\hline
		( & \subst(\incls),                                             & H\lconcat \subst(\frontier),
		  & \lass,
		  & \subst(\solvedeqs) \cup \{ \equationNodeArgs{x}{\sterm} \}, & \lenvars
		  & )
	\end{array}\
	\varphi_{\ruleLLSubs}
\]
\vspace*{-5mm}
\begin{align*}
	\textstyle
	\text{where }
	\varphi_{\ruleLLSubs} & \defiff \subst = \{x \mapsto \sterm\} \text{, } \varsOf{x \equals \sterm} \subseteq \lenvars {, } \lass(\sterm) \subseteq \lass(x)\text{, and} \\
	                      & H = \subst(\{T(y,x) \mid y \in \lenvars, |\sterm| > 1\})\setminus \subst(\frontier)\text{.}
\end{align*}
}%
We apply the rule only if $x$ is length-aware variable (thus so are all variables of $\sterm$), since non-length-aware $x$ is better handled by $\ruleSkip{}$.
Furthermore, as we have directly chain-free formula, $x \equals \sterm$ can be the only constraint with $x$ in input, therefore, we only need to put transducer constraints that stop being concatenation-free to frontier.



\scaledvspace{-1mm}
\subsection{Generation of the Length Image}
\scaledvspace{-1mm}
\label{sec:lengths}

Finally, the rule $\ruleStable{}$ is applied if the frontier is empty.
Following \cref{lemma:stable-solved}, we infer the length image of the concatenation-free part
$(\varphi_m(\incls, \lenvars) \land \lass \land \solvedeqs)$ where $\varphi_m(\incls, \lenvars)$ contains all (concatenation-free) transducer constraints of $\incls$ containing a length-aware variable from $\lenvars$ on input.
If the length image with the original length constraint $\constraintsLengths$ is satisfiable, the whole formula is satisfiable and we return $\sat$.
	{
		\small
		\[
			\ruleStable{}:
			\begin{array}{ccccccc}
				( & \incls, & \emptyset, & \lass, & \solvedeqs, & \lenvars & )
				\\
				\hline
				\multicolumn{7}{c}{\sat}
			\end{array}\
			\lenimg{\varphi_m(\incls, \lenvars) \land \lass \land \solvedeqs} \land \constraintsLengths \text{ is satisfiable}\text{.}
		\]
	}

In order to implement \ruleStable{}, we need to generate the length image $\lenimg\cube$ of a cube
$\cube$.
%
The technique is standard and closely follows \cite{ChainFree}.
The first part is the formula
$\lenimg{\Eqof\cube} = \bigwedge_{\equationNodeArgs{x}{x_1\cdots x_m} \in \Eqof\cube} {\lenvar x}={\lenvar{x_1} + \cdots + \lenvar{x_m}}$
that encodes the binding equations.
This is conjoined encoded with the length image $\lenimg{\Trof\cube \cap \Regof\cube}$,
obtained by constructing a multi-tape transducer constraint equivalent to $\Trof\cube \cap \Regof\cube$ and computing its Parikh image~\cite{Parikh66}.
To construct the multi-tape transducer, assume that  $S_1(\cdot, \cdot), \dots, S_k(\cdot, \cdot)$ are the constraints of $\Trof\cube$ in the topological order wrt to their inclusion graph.
%
The resulting multi-tape transducer $T_k(x_1,\ldots,x_n)$ has one tape for each variable in $\varsOf{\Trof\cube}$ and it is computed inductively, with $T_1 = S_1$ and $T_{\ell + 1}(x_1,\ldots,x_k, x_{k+1})$ computed from $T_{\ell}(x_1,\ldots,x_k)$ and $S_{\ell+1}(x_{k+1},x_i)$ where $1\leq i \leq k$ as
$
	T_{\ell + 1} = \syncOnTapes{T_{\ell} }{S_{\ell+1}}{i}{2}.
$
%
%
%
Due to the chain-freeness, we can be sure that $x_{k+1}$ is not on the input side of any $S_j(\cdot, \cdot),j\neq k$.
%
The newly obtained multi-tape transducer constraint
$T_{\ell + 1}(x_1,\ldots,x_{k+1})$ is then equivalent to $S_{\ell+1}(x_{k+1},x_i) \land T_{\ell}(x_1,\ldots,x_k)$.
After the resulting multi-tape transducer $T_k$ is synchronized with the regular constraints (by a composition of $T_k$ on $i$-the tape with $\Regof\cube(x_i)$),
the final length image is derived from its Parikh image (where each tape corresponds to a single variable).

To keep the multi-tape transducer and its Parikh image compact, we use the following observation: since the Parikh image is only needed to generate the length image, we can abstract away from the differences between symbols of $\Sigma$, and replace every symbol from $\Sigma$ on the transitions of $T_k$ with the same dummy symbol.
The length image will not change.
Then, we apply the simulation-based
reduction making the transducer often significantly smaller and hence the formula becomes much smaller and faster
to solve.


\begin{wrapfigure}[6]{r}{0.25\textwidth}
	\centering
	\vspace{-10mm}
	\begin{tikzpicture}[>=stealth',shorten >=0pt,auto,node distance=22mm,scale=0.8,
  every state/.style={minimum size=5mm,inner sep=1pt}, transform shape]
  \node[state, accepting, initial, initial text={}] (q0) {$q_0$};
  \node[state, right of=q0] (q1) {$q_1$};

  \draw[->] (q0) edge node[above]{\small $\#/\#/\#$} (q1);
  \draw[->] (q1) edge[bend left=40] node[below]{\small $\varepsilon/\varepsilon/\#$} (q0);
  \draw[->] (q1) edge[bend right=40] node[above]{\small $\varepsilon/\#/\varepsilon$} (q0);
\end{tikzpicture}
	\caption{$T(x,y,z)$}
\end{wrapfigure}

\beginexample\label{ex:lenimg}
Consider the 3-tape transducer $T(x,y,z)$ over the single dummy symbol $\#$
shown right.
Each transition label $a_1/a_2/a_3$ indicates what is read on tapes $x$, $y$, $z$,
respectively.
Let $n_1, n_2, n_3 \geq 0$ count how many times the lower ($\varepsilon/\varepsilon/\#$), middle ($\#/\#/\#$), and upper
($\varepsilon/\#/\varepsilon$) transition is taken.
Then, the Parikh image of $T$ is given as
$
	\exists n_1, n_2, n_3 \geq 0 \colon
	\lenvar{x} = n_2  \land
	\lenvar{y} = n_2 + n_3\land
	\lenvar{z} = n_1 + n_2 \land n_1 + n_3 = n_2.
$
\sectionEndSymbol

%

\begin{restatable}{theorem}{thmCorrectness}\label{thm:correctness}
	The presented decision procedure terminates on a directly chain-free cube $\formula$, and returns $\sat$ iff $\formula$ is satisfiable.
\end{restatable}
\begin{proof}[Sketch]
	Soundness and completeness are preserved by every rule application: the refinement rules split the current cube into a finite disjunction of existential extensions, while the substitution rules apply an equivalence-preserving substitution and only reinsert constraints whose stability or concatenation-freeness may have been affected. Thus every branch of the worklist search represents a satisfiability-equivalent refinement of the input cube. When the frontier becomes empty, the remaining cube is in the stable-concatenation-free form, so by \cref{lemma:stable-solved} the satisfiability of the current branch is exactly the satisfiability of its length image, which is what the rule \ruleStable{} checks.

	For termination, the crucial invariant is acyclicity of the inclusion graph. The procedure processes constraints in topological order, so in the absence of substitutions it makes only one forward pass through the graph. The only rules that can move work back in the graph are \ruleLSubs{} and \ruleLLSubs{} by reinserting the transducer constraint that became non-concatenation-free after the substitution. However, these reinserted constraints contain only a single variable on the input. Substitutions generated from refining constraints of this form may affect only outputs of constraints that are predecessors in the inclusion graph. Since these predecessors are already processed, they have only a single input variable, and so their processing affects only their predecessors, etc. Therefore reinsertion can only traverse the acyclic graph backwards, and every backward traversal is finite. Combining this with the finite forward pass yields termination of every branch of the search tree, and thus termination of the whole procedure.
\end{proof}


\subsection{Heuristics and Optimizations}\label{sec:heuristics}
\scaledvspace{-0.1cm}

We propose three heuristics and optimizations that appear to work well in our benchmark.

\scaledvspace{-2mm}
\paragraph{Reduction of transducer constraints.}
Consider a case of a transducer constraint where either its input or output consists of \emph{literals} only (variables whose language is a singleton). Then, we can compute the image of languages of the literals and reduce the transducer constraint to a simpler equational constraint, as formalized by the rule:
{
\small
\[
	\rulereducetrans{}:
	\begin{array}{rcccccl}
		( & \incls \uplus \{ T(\sterm, \tterm) \},          & T(\sterm, \tterm)\lconcat \frontier,            & \lass, & \solvedeqs, & \lenvars & )
		\\
		\hline
		( & \incls \cup \{ \equationNodeArgs{f}{\tterm} \}, & \equationNodeArgs{f}{\tterm}\lconcat \frontier,
		  & \lass \cup \{ f \mapsto T[\lass(\sterm)] \},
		  & \solvedeqs,                                     & \lenvars
		  & )
	\end{array}\
	\varphi_{\rulereducetrans{}}
\]
\vspace*{-6mm}
\begin{align*}
	\textstyle
	\text{where }
	\varphi_{\rulereducetrans{}} \defiff &
	|\lass(\sterm)| = 1 \wedge f \text{ is a fresh variable}\text{.}
\end{align*}
}%
In a similar fashion, we define a rule reducing $T(\sterm, \tterm)$
where $|\lass(\tterm)| = 1$ (in that case the language of $f$ is
given as pre-image of $\lass(\tterm)$).

\scaledvspace{-2mm}
\paragraph{Homomorphic transducer constraints.}
For the case that a transducer constraint $T$ is a \emph{homomorphism}, i.e., for each word $w_1$, $w_2$, and $z$
we have $T(w_1w_2, z) \iff \exists z_1, z_2: T(w_1, z_1) \land T(w_2, z_2) \land z = z_1z_2$,
we can transform the transducer constraint to a simpler form in order to reduce the number
of noodles created during the noodlification. Formally, we add the following rule to the
decision procedure:
{
\small
\[
	\rulehomtrans{}:
	\begin{array}{rcccccl}
		( & \incls \uplus \{ T(\sterm_1\sterm_2, \tterm) \}, & T(\sterm_1\sterm_2, \tterm)\lconcat \frontier, & \lass, & \solvedeqs, & \lenvars & )
		\\
		\hline
		( & \incls \cup H,                                   & H\lconcat \frontier,
		  & \lass,
		  & \solvedeqs,                                      & \lenvars
		  & )
	\end{array}\
	\varphi_{\rulehomtrans{}}
\]
\vspace*{-6mm}
\begin{align*}
	\textstyle
	\text{where }
	 & \varphi_{\rulehomtrans{}} \defiff T \text{ is homomorphism, } H = \{ T(\sterm_1, x_1), T(\sterm_1, x_2), \equationNodeArgs{x_1x_2}{\tterm} \}\text{,} \\&x_1 , x_2 \text{ are fresh}\text{.}
\end{align*}
}

\scaledvspace{-3mm}
\paragraph{Intersection emptiness checking.}
Non chain-free cubes of the form
$T_1(x,y) \land T_2(x,y)$ where $x,y$ are variables are not decidable (c.f.~\cite{morvan-rat-graphs}), but we can often decide them, based on a heuristic observing that the string constraint is satisfiable iff $T_1 \cap T_2 \neq \emptyset$.
In the first step, we create a transducer
$T = \projection{-1}(\syncOnTapes{\ft{T}_1}{\ft{T}_2}{1}{1})$. $T$ is a 2-tape transducer with a property that the original constraint is satisfiable iff
$(w,w) \in R(L(T))$ for some $w\in\Sigma^*$. In the second step, we choose a threshold $\ell$ (empirically chosen as $4$ with higher values having no impact due to the structure of the transducer constraints occurring in the benchmark) and explore reachable configurations of $T$ where we track the possible content of both tapes.
If for some configuration the tapes are not prefixes of the other, we can stop the expansion of this
node as we surely cannot reach a configuration with the same string on both tapes. If for some configuration the length of some tape exceeds $\ell$,
we report \emph{unknown}. If we reach an accepting configuration with $(w,w)$ on tapes, the constraint is \emph{sat}. Otherwise, we report \emph{unsat}.

\scaledvspace{-0.1cm}
\section{Experimental Evaluation}
\scaledvspace{-0.25cm}


We implemented the presented procedure (\ziiinoodler in the results) on top of the version 1.3.0 of the string solver \ziiinoodler, and compared it with the s.o.t.a solvers \cvcv~\cite{cvc5} (version~1.3.2), \ziii~\cite{z3} (version 4.15.4), \ostrich~\cite{ostrich2} (version 2.0.1), and the previous version 1.3.0 of \ziiinoodler (included as a baseline; does not support \replaceall{} operation, but can sometimes solve them by removing them using its rewriter).
The tools are run with 120\,s timeout, 8\,GiB memory limit.
We use benchmarks from SMT-LIB'25 release~\cite{SMTCOMP25benchmarks}, tracks \QFS and \QFSLIA{}.
We take all formulae with \replaceall{} operations, giving us 4 benchmark sets:
%
\begin{enumerate*}
	\item \pcp~\cite{RegularPropagationAnthonyJez25,MarkgrafPR} (1000 formulae): formulae encoding the \emph{Post correspondence problem} using nested \replaceall operations where 500 formulae are randomly generated PCP instances and 500 formulae are hard PCP instances selected from~\cite{Lorentz01,Omori25},
	\item \rna~\cite{RegularPropagationAnthonyJez25,MarkgrafPR} (1000 formulae): benchmarks modelling a reverse transcription process inspired by bioinformatics, where an unknown RNA is converted into a DNA string by a series of \replaceall operations that simulate nucleotide base pairing,
	\item \transducerplus~\cite{AnthonyInteger2020} (91 formulae): benchmarks encoding 7 rational relations (toLower, toUpper, trim, escapeString, addSlashes, htmlEscape, and htmlUnescape) using deeply nested \replaceall and asking whether for each two different relation, one of the four operations (commutativity, duality, equivalence, and idempotence) holds, and
	\item \webapp (340 formulae): formulae from real world web applications using regular constraints, word equations, and \replaceall operation. Note that the original set contains 681 formulae, we selected only those containing \replaceall operation.
\end{enumerate*}
Note that the benchmark set does not contain and negated functional transducer constraints, as they are not supported by the SMT-LIB standard and to our knowledge no such benchmarks exist.

We apply the presented procedure to handle SMT-LIB functions \texttt{replace\_all}, \texttt{replace\_re}, and \texttt{replace\_re\_all} operations (and support \texttt{to\_upper} / \texttt{to\_lower} as non-standard extensions).
However, the method works for any transducer-encoded string relation:
for example, string rotations, interleavings, and other functions that currently have no direct SMT-LIB support.
If SMT-LIB were to allow explicit transducer specifications as first-class terms, our procedure would accept and solve those constraints directly.

%

\begin{figure}[t]
	\begin{subfigure}{0.33\linewidth}%
		\centering%
		\includegraphics[
			width=0.8\linewidth,
			alt={Scatter plot comparing running times of Z3-Noodler and cvc5. Each point represents a benchmark instance, with axes showing time (seconds) for each tool. Points along the diagonal indicate similar performance; deviations show differences. Z3-Noodler significantly outperforms cvc5 on all the benchmarks except Transducer+ where cvc5 is only slightly faster.}
		]{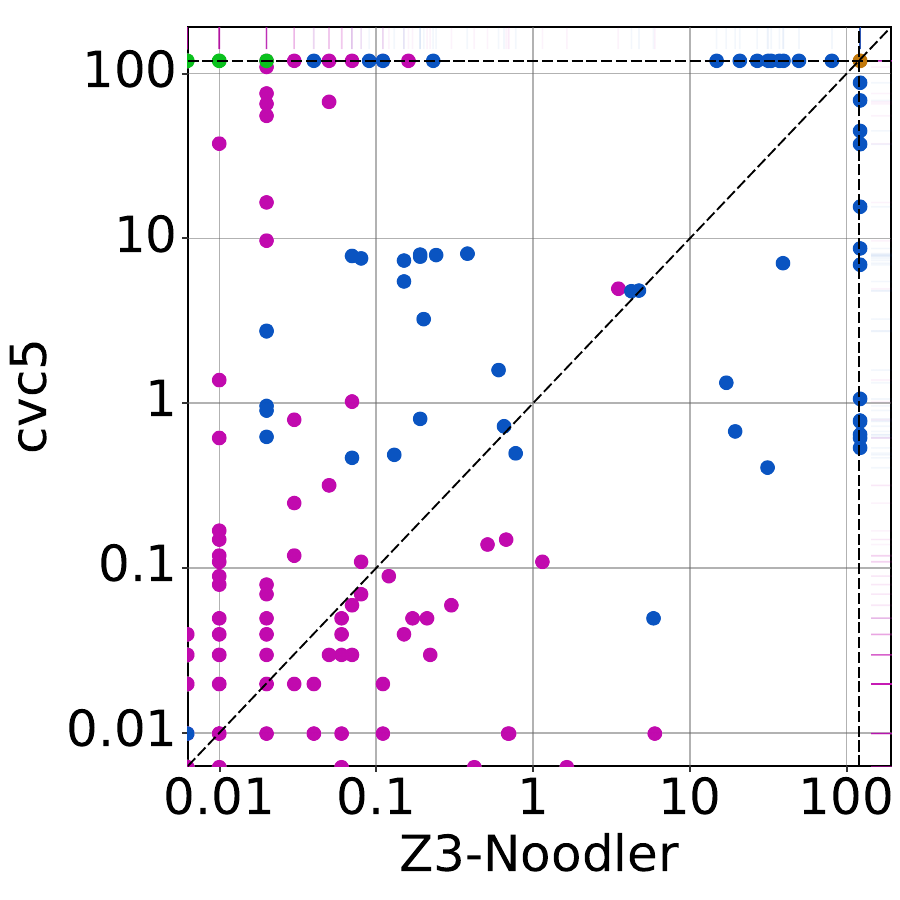}%
	\end{subfigure}%
	\begin{subfigure}{0.33\linewidth}%
		\centering%
		\includegraphics[
			width=0.8\linewidth,
			alt={Scatter plot comparing running times of Z3-Noodler and Z3. Each point represents a benchmark instance, with axes showing time (seconds) for each tool. Points along the diagonal indicate similar performance; deviations show differences. Z3-Noodler significantly outperforms Z3 on all the benchmarks. Z3 timeouts on most of the benchmarks or directly cannot handle the constraints.}
		]{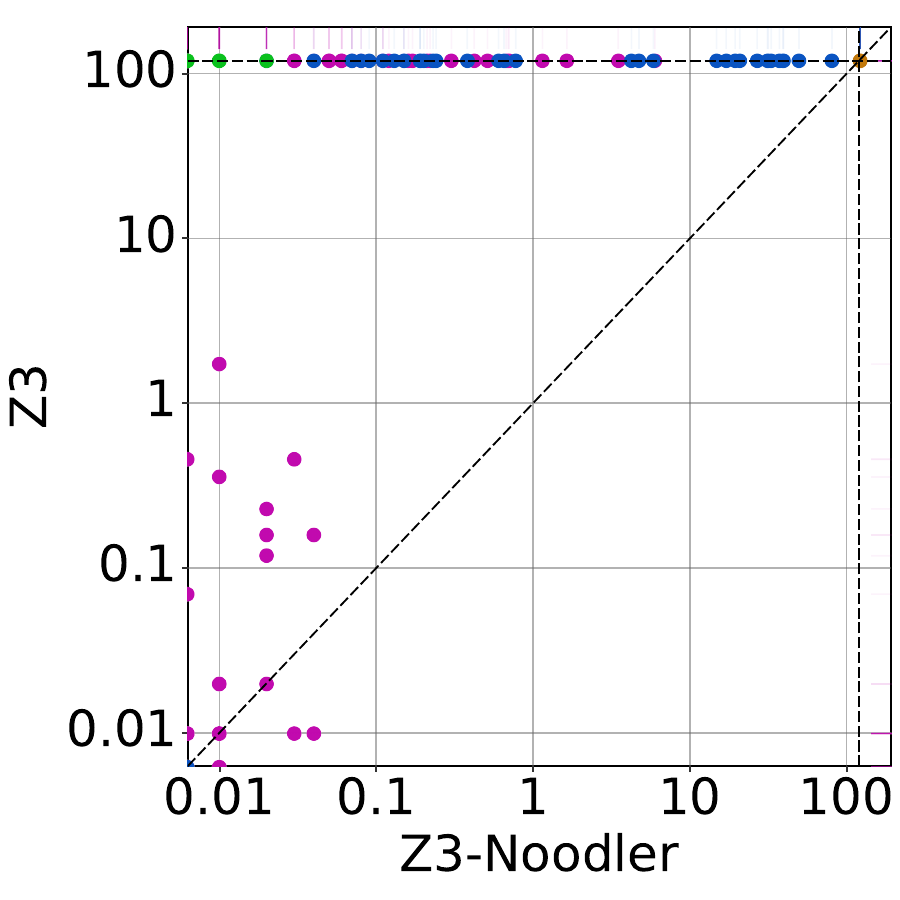}%
	\end{subfigure}%
	\begin{subfigure}{0.33\linewidth}%
		\centering%
		\includegraphics[%
			width=0.8\linewidth,%
			alt={Scatter plot comparing running times of Z3-Noodler and OSTRICH. Each point represents a benchmark instance, with axes showing time (seconds) for each tool. Points along the diagonal indicate similar performance; deviations show differences. Z3-Noodler significantly outperforms OSTRICH on all the benchmarks except a few benchmark instances.}%
		]{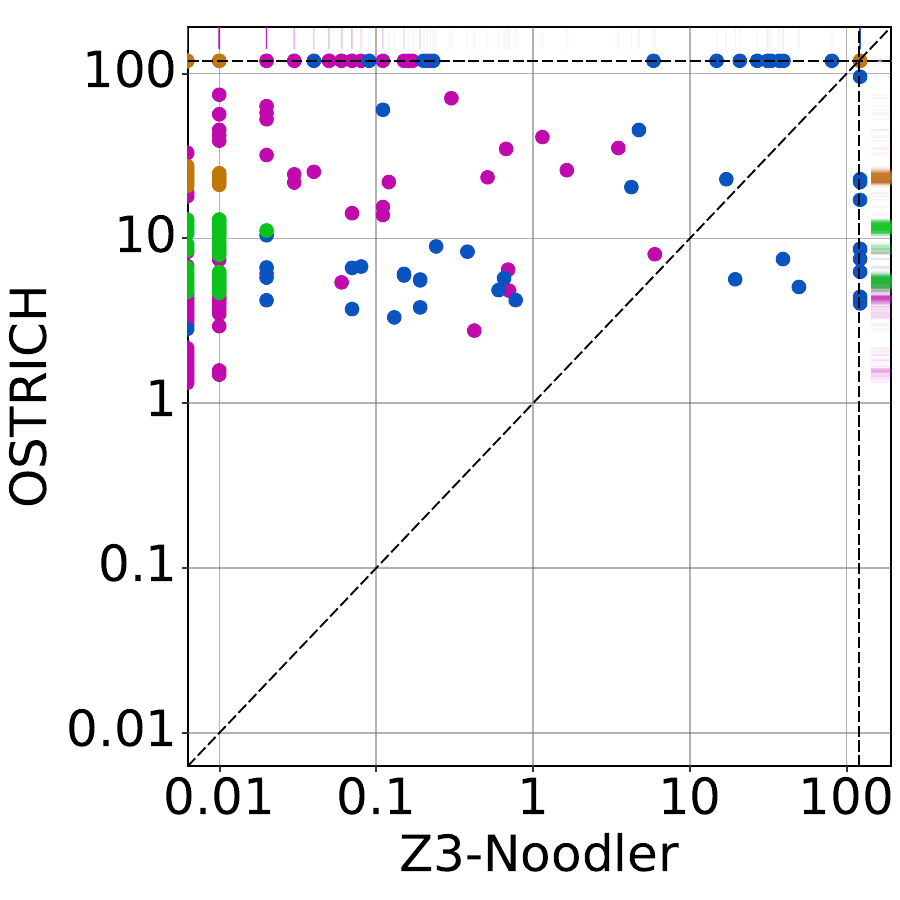}%
	\end{subfigure}%
	\caption{
		Comparison with \cvcv, \ziii, and \ostrich.
		Times are in seconds, axes are logarithmic.
		Colors distinguish benchmark sets:
		\RGBcircle{194,120,11}~\pcp,
		\RGBcircle{35,194,29}~\rna,
		\RGBcircle{13,85,,194}~\transducerplus, and
		\RGBcircle{194,10,175}~\webapp.
	}
	\label{fig:scatter-plots}
	\vspace{-5mm}
\end{figure}

\paragraph{Results.}
The results summarizing the number of solved instances and running times for \ziiinoodler and other tools are given in \cref{tab:comparison}.
In addition, \cref{fig:scatter-plots} shows scatter plots comparing \ziiinoodler with other tools.
%
%
In the table, we also evaluate the impact of the heuristics from \cref{sec:heuristics} where
\ziiinoodlernohomo is a version of \ziiinoodler without the heuristic for the homomorphic transducer constraints, while \ziiinoodlernopre is a version without the intersection emptiness checking heuristic.
We do not evaluate the impact of the heuristic for reducing transducer constraints with string literals as it is tightly integrated into the decision procedure and cannot be easily switched off.

\begin{table}[t]
	\caption{
		The number of solved instances and the time without timeouts (seconds) needed to solve them per tool and benchmark set.
		%
		%
		The size of each benchmark set is in parentheses,
		\nearzero is a~value close to zero.
		%
	}
	\label{tab:comparison}
	\resizebox{\linewidth}{!}{%
		\begin{booktabs}{colspec={lQ[r,gray!30]Q[r,gray!30]rrQ[r,gray!30]Q[r,gray!30]rrQ[r,gray!30]Q[r,gray!30]},cell{1,2}{even}={c=2}{c}} 
\toprule
& \pcp & & \rna & & \transducerplus & & \webapp & & All & \\
& (1,000) & & (1,000) & & (91) & & (340) & & (2,431) & \\
\cmidrule[lr]{2-3}\cmidrule[lr]{4-5}\cmidrule[lr]{6-7}\cmidrule[lr]{8-9}\cmidrule[lr]{10-11}
& solved & time & solved & time & solved & time & solved & time & solved & time \\
\midrule

\ziiinoodler   & \br{500} & \nearzero & \br{1000} & 8     & 41      & 455       & 338      & 20        & \br{1879} & 482       \\
\ziiinoodlernohomo & \br{500} & \nearzero & \br{1000} & 8     & 36      & 497       & \br{339} & 22        & 1875      & 527       \\
\ziiinoodlernopre  & 0        & N/A       & \br{1000} & 8     & 41      & 477       & 338      & 19        & 1379      & 503       \\
\cvcv              & 0        & N/A       & 0         & N/A   & \br{42} & 372       & 293      & 456       & 335       & 827       \\
\ostrich           & 479      & 11,276    & \br{1000} & 8,530 & 39      & 613       & 247      & 1,913     & 1765      & 22,332    \\
\ziii              & 0        & N/A       & 0         & N/A   & 1       & \nearzero & 213      & 5         & 214       & 5         \\
\ziiinoodlerold    & 0        & N/A       & 0         & N/A   & 0       & N/A       & 25       & \nearzero & 25        & \nearzero \\
\bottomrule
\end{booktabs}
	}
	\vspace{-6mm}
\end{table}

For the \pcp and \rna benchmark sets, only \ziiinoodler and \ostrich can solve any instance.
For \pcp{}, both \ziiinoodler and \ostrich can solve only the randomly generated instances, where \ziiinoodler solves all of them nearly instantly, while \ostrich takes over 3 hours to solve 479 instances.
Since all the instances from this benchmark set lead to non-chain-free constraints, the intersection emptiness checking heuristic is needed here for \ziiinoodler.
Without it, it is unable to solve any formula.

Both \ziiinoodler and \ostrich solve all instances in the \rna benchmark set but \ostrich is 1000 times slower than \ziiinoodler.
%
%
%
%
\cvcv solves the most instances in the \transducerplus benchmark set, with \ziiinoodler, \ostrich close behind.
Their runtimes are also similar.
\ziii solves only one instance here.
For \ziiinoodler, the heuristic for homomorphic transducers has an impact here.
Without it, \ziiinoodler is slower and solves 5 instances less.
%
%
For \webapp, \ziiinoodler solves more than other tools, nearly all instances, while being much faster. 
Surprisingly, \ziiinoodlernohomo can solve one more instance than \ziiinoodler.
This is not caused by the heuristic, but by different choices within the SAT solver guiding the solving.
\ziiinoodlerold solves 25 instances while not supporting \replaceall and \replacereall operations: the operators are simple enough to be removed by the rewriter without reaching the string solver.

Except on \transducerplus, \ziiinoodler alone yields the same results as a virtual best solver (VBS) composed of all solvers combined (thus, other solvers do not contribute to the VBS).
For \transducerplus, the deep nesting of \replaceall{} operations generates large transducers in \ziiinoodler, sometimes leading to resource exhaustion.
Thus, VBS that includes \ziiinoodler and either \cvcv or \ostrich solves on \transducerplus 12 additional instances compared to \ziiinoodler alone.
When all three solvers are combined in the VBS, it solves 18 additional instances on \transducerplus compared to \ziiinoodler alone (adding \ziii has no impact).
This also demonstrates that our proposed approach and the approach of \cvcv are orthogonal.
However, since the two approaches are quite different, combining them would be highly complicated.
%
%
Overall, \ziiinoodler consistently outperforms other solvers in both the number of solved instances and runtime, except on one benchmark set where it is a close second.


\scaledvspace{-2mm}
\scaledvspace{-0.25cm}
\section{Conclusion}
\scaledvspace{-0.25cm}
We extended the stabilization-based automata approach to the chain-free fragment with transducer constraints, efficiently integrating them with length reasoning while avoiding variable splitting and monadic decomposition. We also introduced several heuristics that proved highly effective in practice, and our procedure extends the framework with support for negations, allowing unrestricted disequalities and negated functional transducer constraints. Implemented on top of \ziiinoodler, our method significantly outperforms existing solvers on benchmarks with relational constraints, often by orders of magnitude.

	{\fontsize{9pt}{11pt}
		\subsubsection{Acknowledgments.}
		This work was supported by
		the Czech Science Foundation projects GA25-17934S and GA26-22640S.
		the FIT BUT internal project FIT-S-26-9011, and
		\raisebox{-1pt}{\protect\includegraphics[height=8pt]{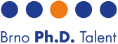}}
		The work of David Chocholatý, Brno Ph.D. Talent Scholarship Holder, is Funded by the Brno City Municipality.

		\subsubsection{Disclosure of Interests.}
		The authors have no competing interests to declare that are relevant to the content of this article.

		\subsubsection{Data-Availability Statement.}
		The source code of \ziiinoodler are available at~\cite{noodlerGithub}.
		The benchmarks, the tools, the evaluation scripts, and the results of the experimental evaluation are available at~\cite{noodlerArtifactCAV26}, to be run on CAV'26 Artifact Evaluation virtual machine~\cite{ArtifactVMOfficialCAV26}.

	}

\bibliographystyle{splncs04}
\bibliography{literature}

\clearpage
\appendix
\section{Proof of \Cref{lemma:chain-free-equivalence}}

\lemmaChainFreeEquivalence*
\begin{proof}[sketch]
	The equivalence of the chain-free definition using an inclusion graph from~\cite{NoodlerFM23} with the original definition of chain-free constraints from~\cite{ChainFree} was shown in~\cite[Lemma 11]{NoodlerFMJournal25}.
	Since the definition of chain-freeness in~\cite{ChainFree} is based on a notion of a splitting graph which, along with the inclusion graph from~\cite{NoodlerFMJournal25}, is independent of the type of constraints in the nodes (inclusions or transducers), the equivalence proof from~\cite{NoodlerFMJournal25} still holds.
	Thus, the proof can be easily adapted to our extended definition of the inclusion graph with transducer constraints.
	\qed
\end{proof}

\section{Proof of \Cref{lemma:chain-free-with-negations-and-transducers}}
\lemmaChainFreeWithNegationsAndTransducers*
\begin{proof}[sketch]
	A cube satisfying the conditions of the Theorem can be transformed into a disjunction of positive directly chain-free cubes using the presented encodings in \cref{sec:chain-free} (with possibly extra $\tocode$ functions).
	Each of these cubes can be solved by the decision procedure of \cref{sec:dec-proc}, with special handling of $\tocode$ function.
	The $\tocode$ function is ignored in the decision procedure, until we check for the applicability of the rule $\ruleStable$.
	The formula $\lenimg{\varphi_m(\incls, \lenvars) \land \lass \land \solvedeqs} \land \constraintsLengths$ in the applicability check is extended to $\lenimg{\varphi_m(\incls, \lenvars) \land \lass \land \solvedeqs} \land \constraintsLengths \land \varphi_{\tocode}(\incls, \langAssign, \solvedeqs)$, where $\varphi_{\tocode}(\incls, \langAssign, \solvedeqs)$ is the LIA formula encoding the possible results of used $\tocode$ functions using the approach of~\cite{NoodlerSAT24} (the results of $\tocode$ functions applied on variables used in transducers can be extracted from Parikh image, see \cref{sec:lengths}).
	\qed
\end{proof}

\section{Proofs of \Cref{lemma:stable-has-solution,lemma:stable-projection,lemma:stable-solved}}
\label{sec:app-proofs}

\lemmaStableHasSolution*
\begin{proof}[sketch]
	Let $\formula$ be a stable and directly chain-free cube with all languages of variables non-empty (i.e. regular constraints are satisfiable).
	Let $G_{\formula}$ be the (acyclic) inclusion graph of $\formula$.
	Since $\formula$ is stable, $\langAssign$ satisfies all constraints in $G_{\formula}$ (all inclusions hold). 
	We can therefore construct a solution by first selecting values from $\langAssign$ for all variables that occur only in outputs of constraints in $G_{\formula}$ and then, by traversing the graph in reversed topological order, propagating the values of output variables to input variables (as the inclusions hold, each solution of the output must be the solution of the input and, from chain-freeness, we know all the variables in input are different).
	\qed
\end{proof}

\stableProjection*
\begin{proof}
	We need to show that any assignment $\ass$ to $\sharedv$ with $\ass(x)\in\Regof\cube(x)$ for all $x\in \sharedv$, is a part of some solution of $\sharedv$.
	In the proof of \cref{lemma:stable-has-solution} we show exactly this, from an assignment to output-only variables we construct a solution for a stable cube.
\end{proof}

\stableConcatenationFree*
\begin{proof}
	From \cref{lemma:stable-projection}, we can construct a solution from assignment to output-only variables of $\stof{\cube}$.
	As the variables of $\solvedcOf{\cube}$ and $\Lenof{\cube}$ occur only in the outputs of $\stof{\cube}$, we can construct a solution of $\stof{\cube}$ from the solution of $\solvedcOf{\cube} \land \Regof\cube \land \Lenof\cube$, therefore $\cube$ is an existential extension of $\solvedcOf{\cube} \land \solvedcReg \land \Lenof\cube$.
\end{proof}

\section{Proof of \Cref{lemma:monadic-existential-extension}}
\lemmaMonadicExistentialExtension*
\begin{proof}[sketch]
	Let $\formula = \constraintsRegulars \land \transductionArgs{s}{t}$ be the input formula.
	By construction, $\noodlifyMonadicFt(\formula)$ gives a solution iff $\formula$ has a solution.
	That is, $\noodlificationNftLang$, after delimiter symbols are projected away, preserves solutions of $\transductionArgs{s}{t}$.
	When splitting $\noodlificationNftLang$ into noodles, a union of noodles is equivalent to $\noodlificationNftLang$ (no solution is lost, and no new solution is added).
	Finally, the transformation of a noodle $\noodle$ into a conjunction $\constraintsRegulars_\noodle \land \aligneqin_\noodle \land \aligneqout_\noodle \land \constraintsTransducers_\noodle$ also preserves solutions, as
	\begin{enumerate*}[label=(\roman*)]
		\item concatenating all the segment transducers gives precisely the original constraint,
		\item concatenating all the right sides of the input binding equations gives precisely the input side of the original constraint,
		\item concatenating all the left sides of the output binding equations gives precisely the output side of the original constraint, and
		\item substituting the input and output concatenated variables in the original constraint by the right and left sides of the input and output binding equations, respectively, gives precisely the original constraint $\transductionArgs{s}{t}$.
	\end{enumerate*}
	We at most restrict the existing languages of the variables of $\formula$ by eliminating words that do not satisfy the constraint $\transductionArgs{s}{t}$.
	Indeed, every solution of $\noodlifyMonadicFt(\formula)$ assigns values to all the variables of $\formula$, which at the same time satisfy $\formula$.
	We can clearly see that $\noodlifyMonadicFt(\formula)$ preserves solutions of $\formula$ and adds no new solution, and that $\noodlifyMonadicFt(\formula)$ contains fresh variables.
	If all fresh variables are existentially quantified, we get precisely an existential extension of $\formula$.
	Moreover, $\noodlifyMonadicFt(\formula)$ is clearly in the almost \concatenationFree form, as for each noodle, all inclusions holds, each transducer constraint in $\constraintsTransducers_i$ is concatenation-free transducer constraint by construction, and each input/output binding equation in $\aligneqin_i$ and $\aligneqout_i$ satisfies the definition for almost concatenation-free form.
	\qed
\end{proof}

\begin{restatable}{lemma}{lemmaPreserveModels}\label{lemma:preserve-models}
	For each node of the proof graph $N$ we have that the formula \\$\bigvee_{N' \text{ is successor of }N}\varphi_{N'}$ is an existential extension of $\varphi_N$.
\end{restatable}
\begin{proof}[sketch]
	Consider the rule $\ruleAlignSplitLen{}$. Observe that each successor contains a
	single disjunct from $\noodlifyMonadicFt(\eta \land \lass)$ (and the original constraint is removed). Therefore, the disjunction of the successors is an existential extension of $\varphi_N$ according
	to \cref{lemma:monadic-existential-extension}. Similarly, the lemma for
	the rule $\ruleAlignSplit$ follows from~\cref{sec:noodlestable}. The remaining rules have only a single
	successor and perform equivalent transformations.
	\qed
\end{proof}

\section{Non-chain-free Formulae}

In \cref{sec:chain-free}, we described how to reason about chain-free formulae.
In this section, we briefly show a counterexample demonstrating that our technique cannot be directly applied to non-chain-free formulae without modifications.

\begin{wrapfigure}[10]{r}{0.15\textwidth}
	\vspace{-10mm}
	\centering
	$$
		\begin{array}{c}
			[ab]^*       \\
			\hline
			x            \\
			\hline
			\alpha_{0,0} \\
			T_{0,0}      \\
			\beta_{0,0}  \\
			\hline
			x            \\
			\hline
			[ab]^*
		\end{array}
	$$
	\vspace{-0.7cm}
	\captionof{figure}{
		Elements of $\psi_p$.
	}
	\label{fig:running-example/possible-alignment-non-chain-free}
\end{wrapfigure}

\beginexample
\label{example:non-chain-free}
Consider the formula $\formula \colon \ft{T}(x,x) \land x \in \{ ab, ba \} \land \constraintsLengths$ where $\ft{T}$ is a transducer recognizing the relation $\{(ab, ba), (ba, ab)\}$, and $\constraintsLengths$ are some length constraints where $x$ is a length variable.
Notice that formula $\formula$ has no solution since both sides of $\ft{T}$ comprise of only $x$, hence need to be equal, which contradicts the relation recognized by $\ft{T}$.
However, if we construct $\noodlificationNftLang$ and derive the single possible noodle $\noodle_0$ to refine the languages of $x$, transforming the formula into the monadic form, we get
$\psi_0 = \constraintsRegulars_0 \land \aligneqin_0 \land \aligneqout_0 \land \constraintsTransducers_0 $ where
\begin{enumerate*}
	\item $\constraintsTransducers_0 = \ft{T}_{0,0}(\alpha_{0,0},\beta_{0,0})$ is segment (transducer) constraint where $\alpha_{0,0}, \beta_{0,0}$ are fresh variables, and $\ft{T}_{0,0} = \noodlificationNftLang = \ft{T}$ recognizes the same language $\{ (ab, ba), (ba, ab) \}$ still,
	\item $\aligneqin_0 =  \equationNodeArgs{x}{\alpha_{0,0}}$ are input binding equations,
	\item $\aligneqout_0 = \equationNodeArgs{x}{\beta_{0,0}}$ are output binding equations, and
	\item $\constraintsRegulars_0 = \alpha_{0,0} \in \projTo{1}{\ft{T}_{0,0}} \land \beta_{0,0} \in \projTo{2}{\ft{T}_{0,0}}$ are new regular constraints where $\projTo{1}{\ft{T}_{0,0}} = \projTo{2}{\ft{T}_{0,0}} = \{a,b\}^*$ (the languages remain unmodified).
\end{enumerate*}
The constraint is in monadic form with no variable having an empty language and length image, which disregards the position of the letters in the words, so words $ab$ and $ba$ are equal, having a solution of length 2 (hence, $\formula$ should have a solution), 
yet,
no solution exists.
\sectionEndSymbol

\end{document}